\documentclass[10pt,journal]{IEEEtran}
\usepackage[dvips]{graphicx}
\usepackage{subfigure}
\usepackage{mdwlist}
\usepackage{amsmath}
\usepackage{amsthm}
\usepackage[linesnumbered,ruled]{algorithm2e}
\usepackage{cite}
\usepackage{amssymb}
\usepackage{multirow}
\usepackage{makecell}
\usepackage{enumerate}
\usepackage{tabularx}

\newtheorem{definition}{Definition}

\newtheorem{theorem}{Theorem}
\newtheorem{lemma}{Lemma}
\newtheorem{corollary}{Corollary}

\newtheorem{proposition}{Proposition}
\newtheorem{example}{Example}
\SetAlFnt{\small}
\begin{document}

\title{Frugal Online Incentive Mechanisms for Crowdsourcing Tasks Truthfully}
\author{Dong~Zhao,
        Huadong~Ma,~\IEEEmembership{Member,~IEEE,}
        and~Liang~Liu
\IEEEcompsocitemizethanks{\IEEEcompsocthanksitem D. Zhao, H.-D. Ma, and L. Liu are with the Beijing Key Lab of Intelligent Telecomm.
Software and Multimedia, Beijing University of Posts and Telecomm., Beijing, 100876, China. E-mail: \{dzhao, mhd\}@bupt.edu.cn, liangliu.bupt@gmail.com.}
}

\maketitle

\begin{abstract}
Mobile Crowd Sensing (MCS) is a new paradigm which takes advantage of pervasive smartphones to efficiently collect data, enabling numerous novel applications.
To achieve good service quality for a MCS application, incentive mechanisms are necessary to attract more user participation.
Most of existing mechanisms apply only for the \emph{offline} scenario where all users' information are known a priori.
On the contrary, we focus on a more realistic scenario where users arrive one by one \emph{online} in a random order.
Based on the \emph{online auction} model, we investigate the problem that users submit their private profiles to the crowdsourcer when they arrive, and the crowdsourcer aims at selecting a subset of users before a specified deadline for minimizing the total payment while a specific number of tasks can be completed.
We design three \emph{online mechanisms}, \emph{Homo-OMZ}, \emph{Hetero-OMZ} and \emph{Hetero-OMG}, all of which can satisfy the \emph{computational efficiency}, \emph{individual rationality}, \emph{cost-truthfulness}, and \emph{consumer sovereignty}.
The \emph{Homo-OMZ} mechanism is applicable to the homogeneous user model and can satisfy the \emph{social efficiency} but not \emph{constant frugality}.
The \emph{Hetero-OMZ} and \emph{Hetero-OMG} mechanisms are applicable to both the homogeneous and heterogeneous user models, and can satisfy the \emph{constant frugality}.
Besides, the \emph{Hetero-OMG} mechanism can also satisfy the \emph{time-truthfulness}.
Through extensive simulations, we evaluate the performance and validate the theoretical properties of our online mechanisms.
\end{abstract}


\section{Introduction}
\label{sec:introduction}
Crowdsourcing is a distributed problem-solving model in which a crowd of undefined size is engaged to solve a complex problem through an open call \cite{chatzimilioudis2012crowdsourcing}.
Nowadays, the proliferation of smartphones provides a new opportunity for extending existing web-based crowdsourcing applications to a larger contributing crowd, making contribution easier and omnipresent.
Furthermore, today's smartphones are programmable and come with a rich set of cheap powerful embedded sensors, such as GPS, WiFi/3G/4G interfaces, accelerometer, digital compass, gyroscope, microphone, and camera.
The great potential of mobile phone sensing offers a variety of novel, efficient ways to opportunistically collect data, enabling numerous \emph{mobile crowd sensing} (MCS) applications, such as Sensorly \cite{website:Sensorly} for constructing cellular/WiFi network coverage maps, SignalGuru \cite{koukoumidis2011signalguru}, Nericell \cite{mohan2008nericell} and VTrack \cite{thiagarajan2009vtrack} for providing traffic information, Ear-Phone \cite{rana2010ear} and NoiseTube \cite{stevens2010crowdsourcing} for making noise maps, and LiFS \cite{yang2012locating} for indoor localization.
For more details, we refer interested readers to several survey papers \cite{lane2010survey,ganti2011mobile,chatzimilioudis2012crowdsourcing}.

Adequate user participation is one of the most critical factors determining whether a MCS application can achieve good service quality.
Most of the current MCS applications \cite{website:Sensorly,koukoumidis2011signalguru,mohan2008nericell,thiagarajan2009vtrack,rana2010ear,stevens2010crowdsourcing,yang2012locating} are based on voluntary participation.
While participating in a MCS campaign, smartphone users consume their own resources such as battery and computing power, and expose their locations with potential privacy threats.
Thus, incentive mechanisms are necessary to provide participants with enough rewards for their participation costs.
Most of existing mechanisms \cite{danezis2005much,lee2010sell,duan2012incentive,yang2012crowdsourcing,jaimes2012location} apply only for the \emph{offline} scenario as illustrated in Fig. \ref{fig_offlineScenario}, in which all of participating users report their profiles, including the tasks that they can complete and the bids, to the crowdsourcer (campaign organizer) in advance, and then the crowdsourcer selects a subset of users after collecting the information of all users to achieve a specific objective.
Generally, there are two classes of incentive mechanisms with different objectives: the \emph{budget feasible} mechanisms which aim at maximizing the crowdsourcer's utility (e.g., the total value of all tasks that can be completed by selected users) under a specific budget constraint, and the \emph{frugal} mechanisms which aim at minimizing the crowdsourcer's total payment under the condition that the specific tasks can be completed.
\begin{figure*}[!t]
  \centering{
  \subfigure[Offline scenario]{
    \label{fig_offlineScenario}
    \includegraphics[width=2.1in]{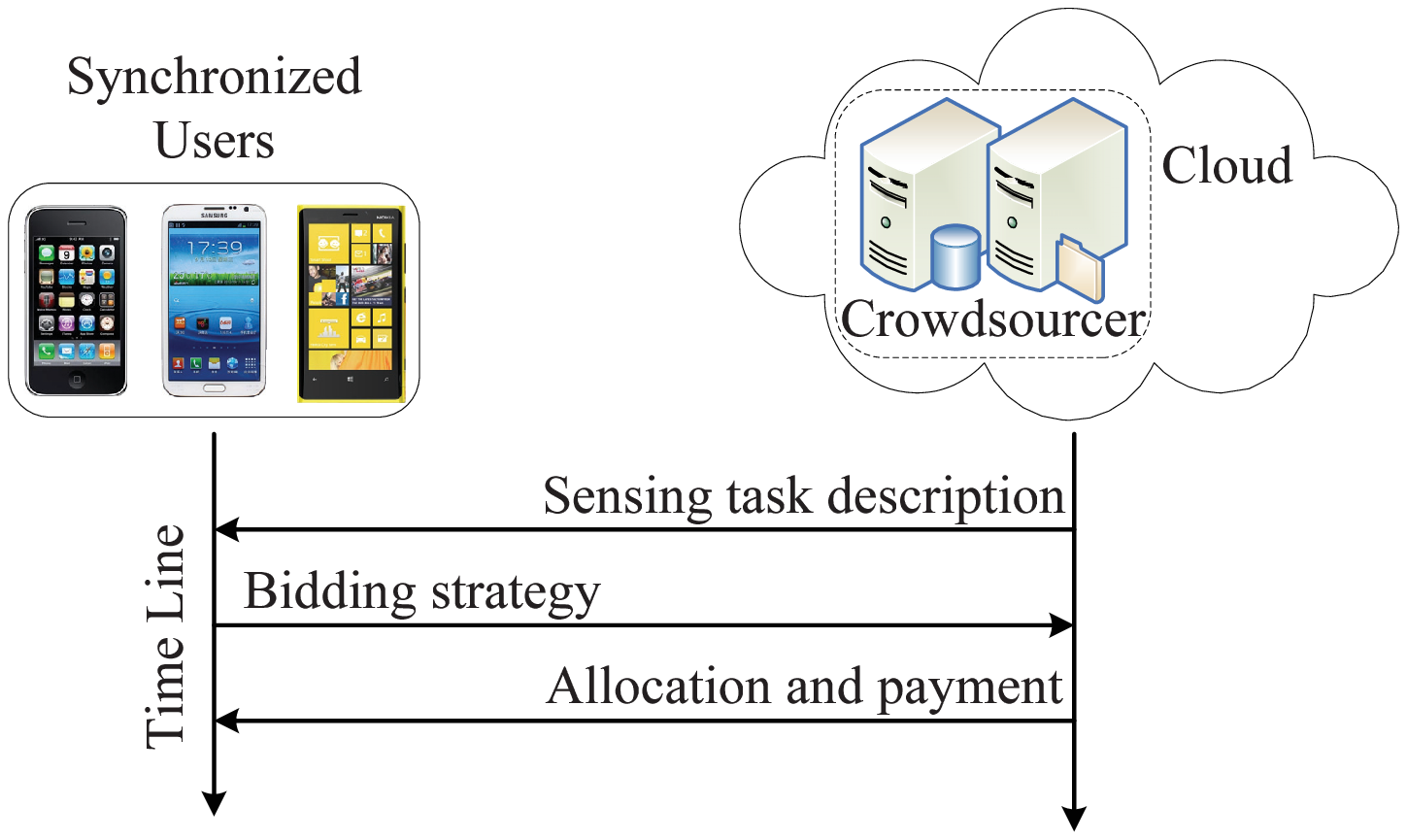}}
  \subfigure[Online scenario: zero arrival-departure interval model]{
    \label{fig_onlineModel1}
    \includegraphics[width=2.0in]{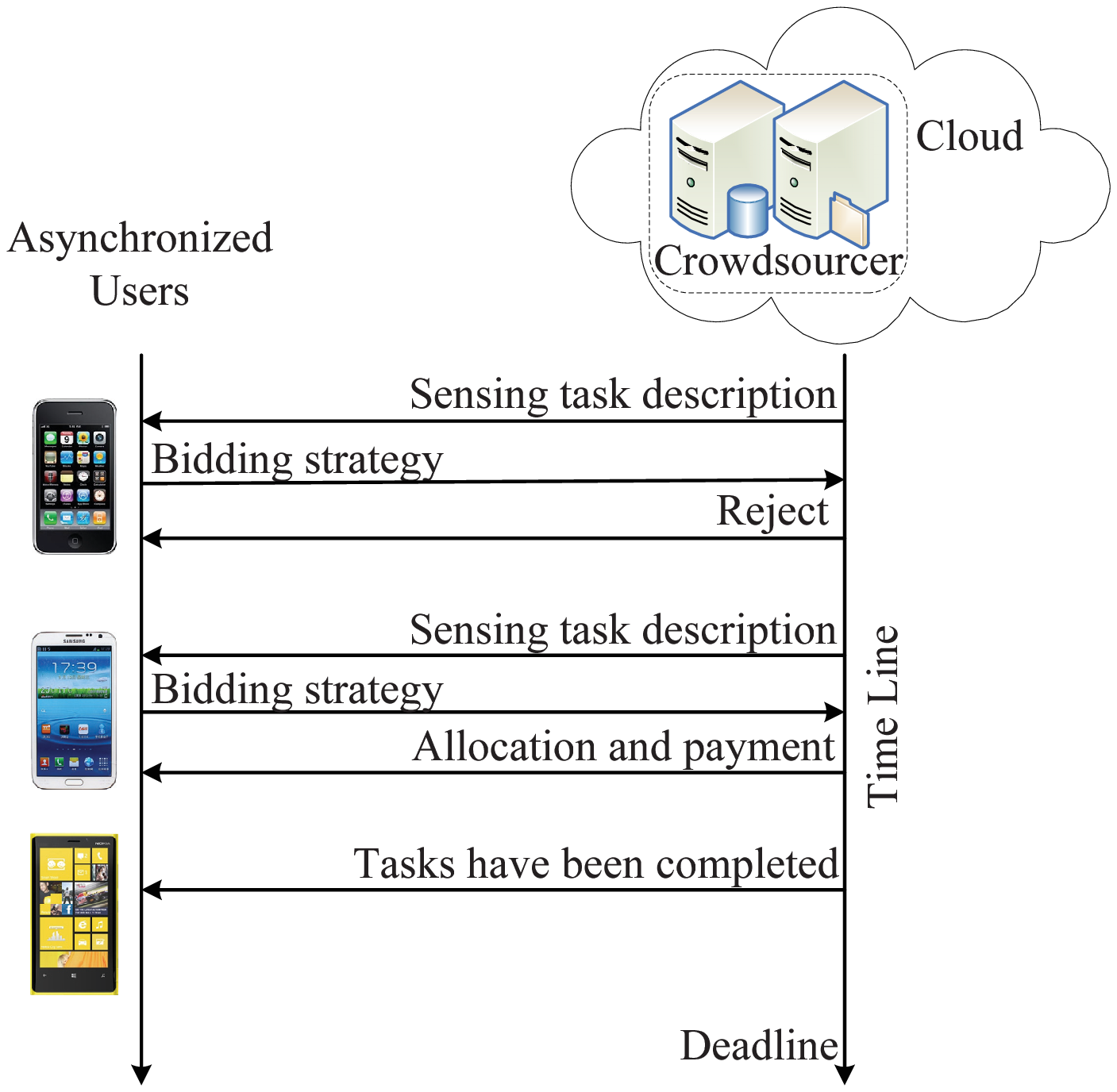}}
    \subfigure[Online scenario: general interval model]{
    \label{fig_onlineModel2}
    \includegraphics[width=2.25in]{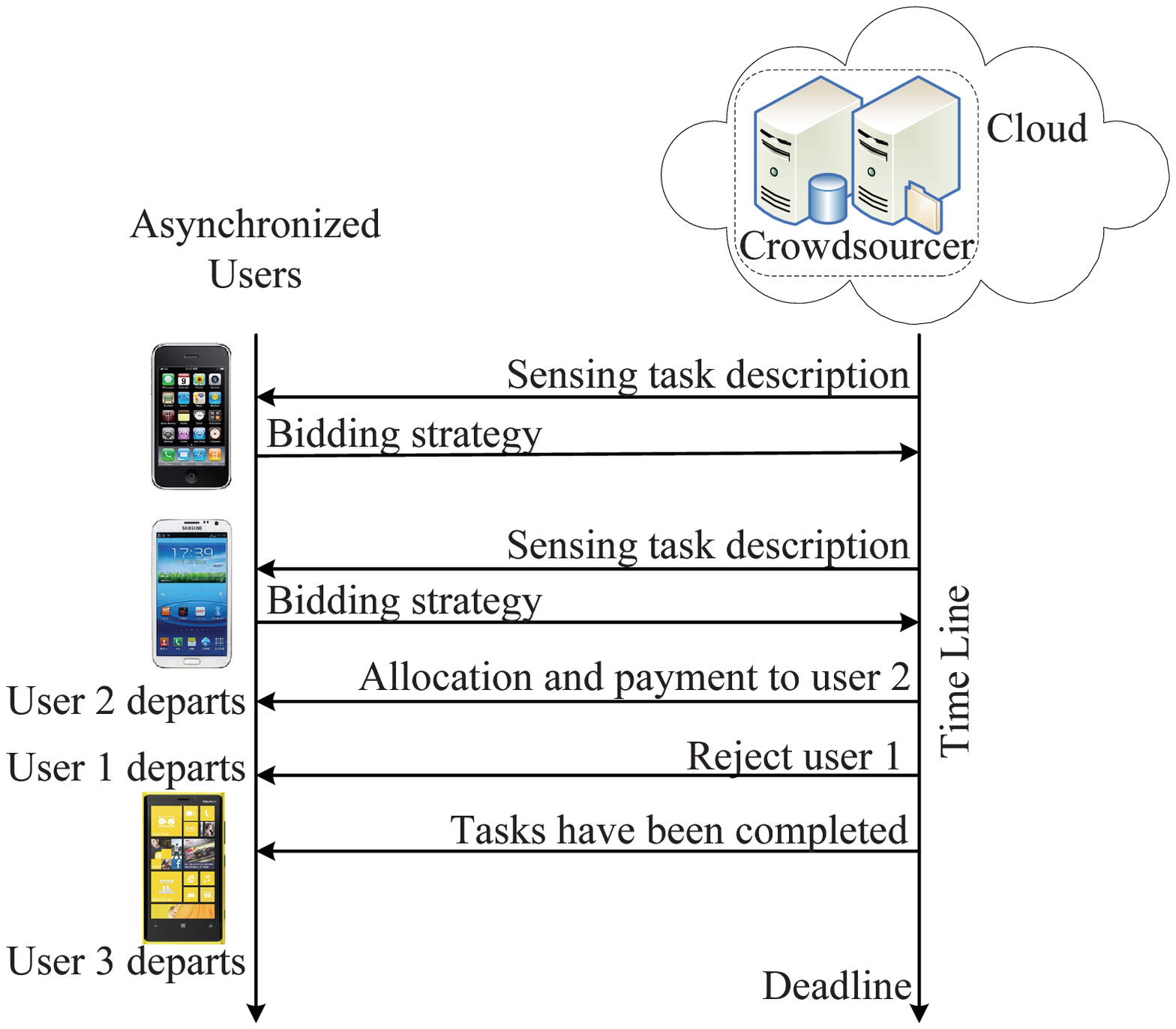}}
  }
  \caption{The comparison of offline and online scenarios for MCS.}
  \label{fig-systemModel}\vspace{-10pt}
\end{figure*}
\begin{table*}[!t]\setlength{\tabcolsep}{2pt}\footnotesize
\begin{center}
\caption{Comparison of various mechanisms proposed in this paper.}
\label{table-comparison}
\begin{tabular}{ccccccccccc}
\hline
    \multirowcell{2}[-.5ex][c]{Mechanism} & homogeneous & heterogeneous & i.i.d. & secretary & computational & individual & cost- & consumer & constant & time-truthfulness\\
    & user model & user model & model & model & efficiency & rationality & truthfulness & sovereignty & frugality & (general interval model)\\
   \hline
   Homo-OMZ & Yes & No & Yes & Yes & Yes & Yes & Yes & Yes & No & NA\\
   Hetero-OMZ & Yes & Yes & Yes & Yes & Yes & Yes & Yes & Yes & Yes & NA\\
   Hetero-OMG & Yes & Yes & Yes & Yes & Yes & Yes & Yes & Yes & Yes & Yes\\
   \hline
   \end{tabular}
\end{center}\vspace{-15pt}
\end{table*}

However, in most of MCS applications users always arrive one by one \emph{online} in a random order and user availability changes over time, which reflects the opportunistic characteristics of human mobility in essence.
For example, in Sensorly \cite{website:Sensorly} or Ear-Phone \cite{rana2010ear}, the crowdsourcer allocates tasks to the smartphone users for sensing Wi-Fi signals or environmental noises only when the users arrive in the area of interest opportunistically.
Therefore, an \emph{online incentive mechanism} is necessary to make irrevocable decisions on whether to accept a user's task and bid \emph{sequentially}, based solely on the information of users arriving before the present moment, without knowing future information, as illustrated in Fig. \ref{fig_onlineModel1} and Fig. \ref{fig_onlineModel2}.

Although the \emph{budget feasible} online incentive mechanisms have been investigated recently \cite{singer2013pricing,singla2013truthful,badanidiyuru2012learning,zhao2014crowdsourcing}, in this paper we focus on the online incentive mechanisms under the framework of \emph{frugality}.
Specially, we consider the following problem: the crowdsourcer aims at selecting a subset of users before a specified deadline, so that the total payment to the selected users can be minimized under the condition that the specific number of tasks can be completed.
We further assume that the cost and arrival/departure time of each user are private and only known to itself.
Users are assumed to be game-theoretic and seek to make strategy (possibly report an untruthful cost or arrival/departure time) to maximize their individual utility in equilibrium.
Thus, the problem of selecting crowdsourcing users while minimizing the total payment can be modeled as an \emph{online auction}.

Our objective is to design online mechanisms satisfying five desirable properties: \emph{computational efficiency}, \emph{individual rationality}, \emph{truthfulness}, \emph{consumer sovereignty} and \emph{constant frugality}.
Informally, \emph{computational efficiency} ensures the mechanism can run in real time, \emph{individual rationality} ensures each participating user has a non-negative utility, \emph{truthfulness} ensures the participating users report their true costs (\emph{cost-truthfulness}) and arrival/departure times (\emph{time-truthfulness}), \emph{consumer sovereignty} ensures each participating user has a chance to win the auction, and \emph{constant frugality} guarantees that the mechanism performs close to the optimal solution in terms of total payment in the offline scenario where all users' information are known a priori.

In this paper, we consider various user models comprehensively from three perspectives: 1) the \emph{zero arrival-departure interval} model (as illustrated in Fig. \ref{fig_onlineModel1}) and the \emph{general interval} model (as illustrated in Fig. \ref{fig_onlineModel2}) with respect to the interval between the arrival time and the departure time of each user, 2) the \emph{homogeneous} user model and the \emph{heterogeneous} user model with respect to the number of tasks that each user can complete, and 3) the \emph{i.i.d.} model and the \emph{secretary} model with respect to the distribution of users.
We propose three online mechanisms, \emph{Homo-OMZ}, \emph{Hetero-OMZ} and \emph{Hetero-OMG}, that are applicable to different models and hold different properties, as shown in Table \ref{table-comparison}.
Note that achieving time-truthfulness is trivial in the \emph{zero arrival-departure interval} model.
The main idea behind our online mechanisms is to adopt a multiple-stage sampling-accepting process.
At every stage the mechanism allocates tasks to a user only if its bid is not less than a certain bid threshold that has been computed using previous users' information, and the number of tasks allocated for the current stage has not been achieved.
Meanwhile, the user obtains a bid-independent payment.
The bid threshold is computed in a manner that guarantees desirable performance properties of the mechanism.

The remainder of this paper is organized as follows.
In Section \ref{sec:problem formulation} we describe the MCS system model, and formulate the problem as an online auction.
We present two online mechanisms, \emph{Homo-OMZ} and \emph{Hetero-OMZ}, under the zero arrival-departure interval model in Section \ref{sec:special case}.
We present the \emph{Hetero-OMZ} mechanism under the general interval model in Section \ref{sec:general case}.
Performance evaluations are presented in Section \ref{sec:performance evaluation}.
We review the related work in Section \ref{sec:related work}, and conclude this paper in Section \ref{sec:conclusion}. 

\section{System Model and Problem Formulation}
\label{sec:problem formulation}
\subsection{System Model}
\label{subsec:systemModel}
As illustrated in Fig. \ref{fig-systemModel}, a MCS system consists of a \emph{crowdsourcer}, which resides in the cloud and consists of multiple sensing servers, and many smartphone \emph{users}, which are connected to the cloud by cellular networks (e.g., GSM/3G/4G) or Wi-Fi connections.
The crowdsourcer first publicizes a MCS campaign, aiming at finding some users to complete a specified number $L\in \mathbb{N}^+$ of tasks.
Then the users interested in participating in the campaign report their profiles to the crowdsourcer.
Finally, the crowdsourcer selects a subset of users to perform tasks for minimizing the total payment based on the profiles of users.
Fig. \ref{fig_offlineScenario} shows an offline scenario, where all of participating users report their profiles to the crowdsourcer synchronously, and then the crowdsourcer allocates tasks to a subset of users by considering the profiles of all users at once.
Unlike the batched and synchronized manner in the offline scenario, the interactive process in the online scenario is sequential and asynchronous, as shown in Fig. \ref{fig_onlineModel1} and Fig. \ref{fig_onlineModel2}.
In this paper we focus on the online scenario.
Specially, we assume that the crowdsourcer requires $L$ tasks to be completed before a specific deadline $T$, and a crowd of users $\mathcal{U}=\{1,2,\ldots,n\}$ interested in participating in the campaign arrive online in a random order, where $n$ is unknown.
Each user $i\in \mathcal{U}$ has an arrival time $a_i\in \{1, \ldots, T\}$, a departure time $d_i\in \{1, \ldots, T\}$, $d_i\geq a_i$, and a number of tasks it can complete $\tau_i$.
Meanwhile, user $i$ also has an associated cost $c_i \in \mathbb{R}^+$ for performing a single task.
All information constitutes the \emph{profile} of user $i$, $\theta_i=(a_i,d_i,\tau_i,c_i)$.
The crowdsourcer, when presented with the profile of user $i$, must decide how many tasks to allocate, and how much payment to pay to user $i$ before it departs.
The decisions should be made one by one until the deadline or all the tasks have been allocated.

\subsection{User Model}
\label{subsec:userModel}
In this paper we consider two models with respect to the interval between the arrival time and the departure time of each user:
\begin{itemize*}
\item \textbf{Zero arrival-departure interval model:} The arrival time of each user equals to its departure time.
\item \textbf{General interval model:} There are no limitations on the arrival-departure interval (i.e., it may be zero or non-zero interval).
\end{itemize*}
The first model is reasonable for the MCS applications where the decision has to be made in time, as illustrated in Fig. \ref{fig_onlineModel1}.
For example, in LiFS \cite{yang2012locating}, the users receive task description when they enter the target building, and report the profiles expecting the platform to reply immediately, since they may not want to be disturbed anymore when they are working or shopping in that building.
Nevertheless, in other scenarios, smartphone users may be patient to wait for the reply of the crowdsourcer for some time interval.
For example, a user who is staying in a traffic tool or drinking at a coffee shop may play with the crowdsourcer for some time.
In such scenarios, the general interval model should be considered where the decision could be put off until when the user departs, as illustrated in Fig. \ref{fig_onlineModel2}.

With respect to the number of tasks that each user can complete, we consider the following two models:
\begin{itemize*}
\item \textbf{Homogeneous user model:} Each user can complete only a single task.
\item \textbf{Heterogeneous user model:} Different users can complete different number of tasks.
\end{itemize*}
Obviously, the homogeneous user model is a special case of the heterogeneous user model, and in the first model, we denote the \emph{profile} of each user $i$ as $\theta_i=(a_i,d_i,1,c_i)$.

With respect to the distribution of users, there are two extremes of the modeling spectrums.
At one extreme is the \emph{oblivious adversarial model}, where an adversary chooses a \emph{worst-case} input sequence including the users' costs, task numbers, and their arrival orders.
This view leads to a very pessimistic result: no online mechanisms can achieve constant competitiveness \cite{bar2002incentive}, but this model is unnecessary to consider in a normal crowdsourcing market.
At the other extreme is a sequence of users with task numbers and the respective unit costs i.i.d. sampled from a \emph{known} distribution.
This view may lead to the optimum strategy by using dynamic programming, but it assumes too much knowledge to be practical.
Thus, two natural candidate models in between the two extremes are considered in this paper:

\begin{itemize*}
\item \textbf{The i.i.d. model:} The numbers of tasks that users can complete and their respective unit costs are i.i.d. sampled from some \emph{unknown} distributions.
\item \textbf{The secretary model:} An adversary gets to decide on the numbers of tasks that users can complete and the respective unit costs, but not on the \emph{order} in which they are presented to the crowdsourcer.
\end{itemize*}
In fact, the i.i.d. model is a special case of the secretary model, since the sequence can be determined by first picking a multi-set of task numbers and costs from the (unknown) distribution, and then permuting them randomly.

\subsection{Online Auction}
\label{subsec:onlineAuction}
We model the interactive process between the crowdsourcer and users as an \emph{online auction}.
Users are assumed to be game-theoretic and seek to make \emph{strategy} to maximize their individual utility in equilibrium.
Each user $i$ makes a \emph{reserve price} $b_i$, called \emph{bid}, based on its valuations on each single task.
Note that the arrival/departure time and cost of user $i$ are private and only known to itself.
Only the number of tasks $\tau_i$ must be true since the crowdsourcer can identify whether the announced tasks are performed.
In other words, user $i$ may misreport all information about its profile except for $\tau_i$.
The crowdsourcer's objective and the deadline $T$ are common knowledge.
Although we do not require a user to declare its departure time until the moment of its departure, we find it convenient to analyze our auctions as direct-revelation mechanisms (DRMs).
The strategy space in an online DRM allows a user to declare some possibly untruthful profile $\hat{\theta_i}=(\hat{a_i},\hat{d_i},\tau_i,b_i)$, subject to $a_i \leq \hat{a_i} \leq \hat{d_i} \leq d_i$.
Note that we assume that a user cannot announce an earlier arrival time or a later departure time than its true arrival/departure time.
In order to complete the required sensing tasks, the crowdsourcer needs to design an \emph{online mechanism} $\mathcal{M}=(f,p)$ consisting of an \emph{allocation} function $f$ and a \emph{payment} function $p$.
For any \emph{strategy sequence} $\hat{\theta}=(\hat{\theta_1},\ldots,\hat{\theta_n})$, the allocation function $f(\hat{\theta})$ computes an allocation of tasks for a selected subset of users $\mathcal{S}\subseteq \mathcal{U}$, and the payment function $p(\hat{\theta})$ returns a vector $(p_1(\hat{\theta}),\ldots,p_n(\hat{\theta}))$ of payments per task to the users.
That is, each selected user $i\in \mathcal{S}$ is allocated $f_i$ tasks at price $p_i$ per task. Thus, in the homogeneous user model, the \emph{utility} of user $i$ is
\[
u_i=
\begin{cases}p_i-c_i,\quad \ \ & \mbox{if}\quad i\in \mathcal{S} ;\\
  0, \quad \ \ & \mbox{otherwise}.
\end{cases}
\]
In the heterogeneous user model, the \emph{utility} of user $i$ is
\[
u_i=
\begin{cases}f_i (p_i-c_i),\quad \ \ & \mbox{if}\quad i\in \mathcal{S} ;\\
  0, \quad \ \ & \mbox{otherwise}.
\end{cases}
\]

Note that, the crowdsourcer, when presented with the strategy $\hat{\theta_i}$ of user $i$, must decide whether to accept user $i$ at what price ($p_i$) before the time step $\hat{d_i}$.
In the homogeneous user model, the crowdsourcer expects that the total payment is minimized under the condition that at least $L$ users are selected before the deadline, i.e.,
\begin{equation}\textbf{Minimize } \sum_{i\in \mathcal{S}}p_i \textbf{ subject to } |\mathcal{S}| = L.\nonumber\end{equation}
In the heterogeneous user model, the crowdsourcer expects that the total payment is minimized under the condition that at least $L$ tasks are completed before the deadline, i.e.,
\begin{equation}\textbf{Minimize } \sum_{i\in \mathcal{S}}f_i p_i \textbf{ subject to } \sum_{i\in \mathcal{S}}f_i = L.\nonumber\end{equation}

Table \ref{table_notations} lists frequently used notations in this paper.
\begin{table}[t]\small
\begin{center}
\caption{Frequently used notations.}
\label{table_notations}
\begin{tabularx}{0.48\textwidth}{c|XcX}
  \hline
  Notation & Description\\
  \hline
  $\mathcal{U},n,i$ & set of users, number of users, and one user\\
  $L,L'$ & total number of tasks and stage-task-number\\
  $T,T',t$ & deadline, end time step of each stage, and each time step\\
  $a_i,\hat{a_i}$ & true arrival time and strategic arrival time of user $i$\\
  $d_i,\hat{d_i}$ & true departure time and strategic departure time of user $i$\\
  $\tau_i$ & the number of tasks that user $i$ can complete\\
  $c_i,b_i$ & true cost and bid of user $i$ for performing a single task\\
  $\theta_i,\hat{\theta_i}$ & true profile and strategy of user $i$\\
  $\mathcal{S},\mathcal{S}'$ & set of selected users and sample set\\
  $f_i,p_i,u_i$ & allocation, payment per task, and utility of user $i$\\
  $b^*$ & bid threshold\\
  $\delta$ & parameter used for computing the bid threshold\\
  $\beta$ & initialized bid threshold\\
  \hline
\end{tabularx}
\end{center}
\end{table}

\subsection{Design Objective}
\label{subsec:objective}
Our objective is to design online mechanisms satisfying the following five desirable properties:
\begin{itemize*}
\item \textbf{Computational Efficiency:} A mechanism is \emph{computationally efficient} if both the allocation and payment can be computed in polynomial time as each user arrives.
\item \textbf{Individual Rationality:} Each participating user will have a non-negative utility: $u_i\geq 0$.
\item \textbf{Truthfulness:} A mechanism is \emph{cost-} and \emph{time-truthful} (or simply called \emph{truthful}, or \emph{incentive compatible} or \emph{strategyproof}) if reporting the true cost and arrival/departure time is a \emph{dominant strategy} for all users. In other words, no user can improve its utility by submitting a false cost or arrival/departure time, no matter what others submit.
\item \textbf{Consumer Sovereignty:} The mechanism cannot arbitrarily exclude a user; each user will have a chance to win the auction and obtain a payment if only its bid is sufficiently low while others are fixed.
\item \textbf{Constant Frugality:} The ideal goal of the mechanism is to minimize the total payment of the crowdsourcer.
Although it is impossible to obtain an optimal solution in the online scenario, we hope that the mechanism can use as small payment as possible for performing all tasks, namely achieving the \emph{frugality}.
Ideally, we would like that the total payment made by the mechanism has a constant approximate ratio compared to the minimal cost required for performing all tasks in the offline scenario where the crowdsourcer has full knowledge about users' profiles.
We call this ideal objective as ``\emph{idealistic frugality}".
However, it is easy to know that no truthful mechanism can perform well even if in the offline scenario\footnotemark[1].
\footnotetext[1]{Consider an instance where a single user can complete all required tasks at a very low cost $\epsilon$, while every other user has a very high cost $C$.
It is not hard to know that any truthful mechanism that completes all tasks must pay at least $C$, thus making the ratio between the total payment and the minimal cost unbounded.}
Therefore, in this paper we use an alternative objective, called as ``\emph{realistic frugality}".
Specially, we use the \emph{frugality ratio} to measure the frugality as follows: given a number of tasks to perform, $L$, we say that a mechanism is $\alpha$-\emph{frugal} if it allocates $L$ tasks in expectation while guaranteeing that the total payment is no more than the minimum cost required to complete $\alpha L$ tasks in the offline scenario.
Here, our goal is to design mechanisms with a constant frugality ratio, namely that the mechanisms can satisfy the \emph{constant frugality}.
\end{itemize*}

The importance of the first two properties is obvious, because they together guarantee that the mechanisms can be implemented in real time and satisfy the basic requirements.
In addition, the last three properties are indispensable for guaranteeing the high performance and robustness.
The \emph{truthfulness} aims at eliminating the fear of market manipulation and the overhead of strategizing over others for the participating users.
The \emph{consumer sovereignty} aims at guaranteeing that each participating user has a chance to win the auction, otherwise it will hinder the users' competition or even result in task starvation.
Besides, if some users are guaranteed not to win the auction, then being truthful or not will have the same outcome.
For this reason, the property satisfying both the \emph{consumer sovereignty} and \emph{truthfulness} is also called \emph{strong truthfulness} \cite{hajiaghayi2004adaptive}.
Later we will show that satisfying the \emph{consumer sovereignty} is not trivial in the online scenario, which is in contrast to the offline scenario.
Finally, we emphasize that the \emph{frugality} is different from \emph{social efficiency} which aims at minimizing the \emph{total cost} (\emph{not total payment}) of selected users. It is generally defined that, an online mechanism is $O(g(n))$-\emph{competitive} for \emph{social efficiency} if the ratio between the total cost achieved by this mechanism and the optimal offline solution is $O(g(n))$.

\section{Online Mechanism under Zero Arrival-departure Interval Model}
\label{sec:special case}
In this section, we consider the zero arrival-departure interval model where each user is impatient since the decision must be made immediately once it arrives.
Note that achieving time-truthfulness is trivial in this model.
It is because that any user has no incentive to report a later arrival time or an earlier departure time than its true arrival/departure time, since the user cannot perform any sensing task or obtain a payment after it departs.
To facilitate understanding, it is also assumed that no two users have the same arrival time.
Note that this assumption can be easily removed according to the revised mechanism in Section \ref{sec:general case}.

\subsection{Homogeneous User Mechanism}
An online mechanism needs to overcome several nontrivial challenges: first, the users' costs are unknown and need to be reported in a truthful manner; second, a specified number of tasks should be completed before the deadline; finally, the mechanism needs to cope with the online arrival of users.
In the homogeneous user model, if we do not consider the truthfulness, then the mechanism design problem becomes a $k$-choice secretary problem in essence.
There are two existing solutions \cite{babaioff2007knapsack,kleinberg2005multiple} that can achieve a constant-factor approximation, and can be applied into online auctions that satisfy the truthfulness.
However, none of two solutions can guarantee the \emph{consumer sovereignty}.
The first solution \cite{babaioff2007knapsack} adopts a two-stage sampling-accepting process: the first batch of users is rejected and used as the sample which enables making an informed decision on whether accepting the rest of users.
In this case, the first batch of users has no chance to win the auction no matter how low its cost is.
It can lead to undesirable effects in our problem: automatically rejecting the first batch of users encourages users to arrive late; in other words, those users arriving early have no incentive to report their bids, which may hinder the users' competition or even result in task starvation.
Although the second solution \cite{kleinberg2005multiple} adopts a multi-stage sampling-accepting process, at the first stage it uses Dynkin's algorithm \cite{dynkin1963optimum} for the classic secretary problem.
Therefore, this solution cannot guarantee the consumer sovereignty, as it is known that Dynkin's algorithm adopts a two-stage sampling-accepting process.

\subsubsection{Mechanism Design}
To address the above challenges, we design an online mechanism, \emph{Homo-OMZ}, by using a \emph{multiple-stage sampling-accepting} process.
The mechanism dynamically increases the sample size and learns a \emph{bid threshold} used for future decision, while increasing the \emph{stage-task-number} it allocates at various stages.
The whole process is illustrated in Algorithm \ref{alg:Homo-OMZ}.
We first divide all of $T$ time steps into $(\lfloor \log_2 T \rfloor+1)$ stages: $\{1, 2, \ldots, \lfloor \log_2 T \rfloor, \lfloor \log_2 T \rfloor+1\}$.
The stage $i$ ends at time step $T'=\lfloor 2^{i-1}T / 2^{\lfloor \log_2 T \rfloor} \rfloor$.
Correspondingly, the stage-task-number for the $i$-th stage is set to be $L'=2^{i-1}L / 2^{\lfloor \log_2 T \rfloor}$, meaning that $L'$ tasks should be allocated before the end time of this stage.
Finally, $L$ tasks in total should be allocated before the deadline $T$.
Fig. \ref{fig-stages} is an illustration when $T=8$.
When a stage is over, we add all users who have arrived into the sample set $\mathcal{S}'$, and compute a bid threshold $b^*$ according to the information of samples and the allocated stage-task-number $L'$.
This bid threshold is computed by calling the \textbf{GetBidThreshold1} algorithm (to be elaborated later), and used for making decision at the next stage.
Specially, when the last stage $i = \lfloor \log_2 T \rfloor+1$ comes, the bid threshold has been computed according to the information of all users arriving before time step $\lfloor T/2 \rfloor$, and the allocated stage-task-number $L/2$.
\begin{figure}[!t]
\centering{
\includegraphics[width=3.2in]{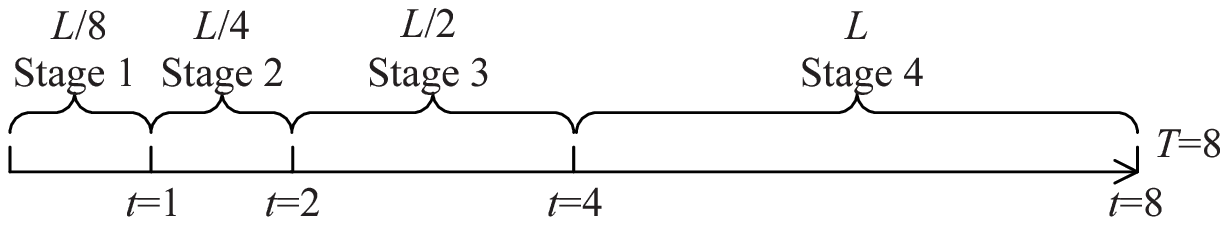}}
\caption{Illustration of a multiple-stage sampling-accepting process when $T=8$.}
\label{fig-stages}
\vspace{-5pt}
\end{figure}

When a new user $i$ arrives, the mechanism allocates a task to it as long as its bid is not less than the current bid threshold $b^*$, and the allocated stage-task-number $L'$ has not been achieved.
Meanwhile, we give user $i$ a payment:
\[p_i=b^*,\]
and add this user to the set of selected users $\mathcal{S}$.
To start the mechanism, we initially set a large bid threshold $\beta$, which is used for making decision at the first stage.

Since each stage maintains a common bid threshold, it is natural to compute the lowest single price that many users will accept from the sample set $\mathcal{S}'$ and the allocated stage-task-number $L'$ as the bid threshold.
Similar to the solution of \cite{kleinberg2005multiple}, we set the bid threshold equal to the $L'$-th lowest bid in the sample set, as illustrated in Algorithm \ref{alg:get threshold1}.

\subsubsection{Mechanism Analysis}
In the following we analyze the properties of the \emph{Homo-OMZ} mechanism from five aspects.

\begin{algorithm}
\SetAlgoLined
\caption{Homogeneous Online Mechanism under Zero Arrival-departure Interval Model (Homo-OMZ)}
\label{alg:Homo-OMZ}
\KwIn{Task number $L$, deadline $T$}
$(t,T',L',\mathcal{S}',b^*,\mathcal{S}) \leftarrow ( 1, \frac{T}{2^{\lfloor \log_2 T \rfloor}}, \frac{L}{2^{\lfloor \log_2 T \rfloor}}, \emptyset, \beta, \emptyset)$\;
\While{$t\leq T$}
{
    \If{there is a user $i$ arriving at time step $t$}
    {
        \eIf{$b_i\leq b^*$ {\bf and} $|\mathcal{S}| < L'$}
        {
            $f_i\leftarrow 1$; $p_i\leftarrow b^*$\;
            $\mathcal{S}\leftarrow \mathcal{S} \cup \{i\}$\;
        }{
            $f_i\leftarrow 0$; $p_i\leftarrow 0$\;
        }
        $\mathcal{S}'\leftarrow \mathcal{S}' \cup \{i\}$\;
    }
    \If{$t=\lfloor T' \rfloor$}
    {
        $b^*\leftarrow \textbf{GetBidThreshold1}(L',\mathcal{S}')$\;
        $T'\leftarrow 2T'$; $L'\leftarrow 2L'$\;
    }
    $t\leftarrow t+1$\;
}
\end{algorithm}

\begin{algorithm}
\SetAlgoLined
\caption{GetBidThreshold1}
\label{alg:get threshold1}
\KwIn{Stage-task-number $L'$, sample set $\mathcal{S}'$}
Sort $\mathcal{S}'$ s.t. $b_1 \leq b_2 \leq \ldots \leq b_{|\mathcal{S}'|}$\;
\Return $b_{L'}$;
\end{algorithm}

\begin{lemma}
\label{lemma:computational efficiency}
The Homo-OMZ mechanism is computationally efficient.
\end{lemma}
\begin{proof}
Since the mechanism runs online, we only need to focus on the computation complexity at each time step $t\in\{1,\ldots,T\}$.
It is easy to know that the running time of computing the allocation and payment of each user $i$ (lines 3-11) is $O(1)$.
Thus, the computation complexity of the mechanism is dominated by the complexity of computing the bid threshold (Algorithm \ref{alg:get threshold1}).
We can use any computationally efficient sorting algorithm to implement the Algorithm \ref{alg:get threshold1}.
For example, the computation complexity of quick sorting algorithm is $O(n\log n)$.
Thus, the Homo-OMZ mechanism is computationally efficient.
\end{proof}

\begin{lemma}
\label{lemma:individual rationality}
The Homo-OMZ mechanism is individually rational.
\end{lemma}
\begin{proof}
From the lines 4-9 of Algorithm \ref{alg:Homo-OMZ}, we know that $p_i\geq b_i$ if $i\in \mathcal{S}$, otherwise $p_i=0$. Thus, we have $u_i\geq 0$.
\end{proof}

Designing a cost-truthful mechanism relies on the rationale of \emph{bid-independence}.
Let $b_{-i}$ denote the sequence of bids arriving before the $i$-th bid $b_i$, i.e., $b_{-i}=(b_1,\ldots,b_{i-1})$.
We call such a sequence \emph{prefixal}.
Let $p'$ be a function from prefixal sequences to prices (non-negative real numbers).
We extend the definition of bid-independence \cite{goldberg2006competitive} to the online scenario:

\begin{definition}[Bid-independent Online Auction]
An online auction is called bid-independent if the allocation and payment rules for each player $i$ satisfy:
\begin{enumerate}[a)]\setlength{\itemindent}{1.1em}
    \item The auction constructs a price schedule $p'(b_{-i})$;
    \item If $p'(b_{-i})\geq b_i$, player $i$ wins at price $p_i=p'(b_{-i})$;
    \item Otherwise, player $i$ is rejected, and $p_i=0$.
\end{enumerate}
\end{definition}
\begin{proposition}
(\cite{bar2002incentive}, Proposition 2.1) An online auction is cost-truthful if and only if it is bid-independent.
\end{proposition}
\begin{lemma}
\label{lemma:cost-truthfulness}
The Homo-OMZ mechanism is cost-truthful.
\end{lemma}
\begin{proof}
Consider a user $i$ that arrives at some stage for which the bid threshold is $b^*$.
If by the time the user arrives there are no remaining tasks unallocated, then the user's cost declaration will not affect the allocation of the mechanism and thus cannot improve its utility by submitting a false cost.
Otherwise, assume there are remaining tasks to be allocated by the time when the user arrives.
In case $c_i \leq b^*$, reporting any cost below $b^*$ would not make a difference in the user's allocation and payment, and its utility would be $b^*-c_i\geq 0$.
Declaring a cost above $b^*$ would make the worker lose the auction, and its utility would be 0.
In case $c_i > b^*$, declaring any cost above $b^*$ would leave the user unallocated with utility 0.
If the user declares a cost lower than $b^*$ it will be allocated.
In such a case, however, its utility will be negative.
Hence the user's utility is always maximized by reporting its true cost: $b_i=c_i$.
\end{proof}

\begin{lemma}
\label{lemma:consumer sovereignty}
The Homo-OMZ mechanism satisfies the consumer sovereignty.
\end{lemma}
\begin{proof}
Each stage is an accepting process as well as a sampling process ready for the next stage.
As a result, users are not automatically rejected during the sampling process, and are allocated as long as their bids are not less than the current bid threshold, and there are still remaining tasks unallocated at the current stage.
\end{proof}

According to the proof of Theorem 2.1 from \cite{kleinberg2005multiple}, it is easy to know that the expected total cost of the $L$ users selected by the \emph{Homo-OMZ} mechanism is at most $1+O(\sqrt{1/L})$ times the sum of the $L$ lowest costs of all users.
Thus, we obtain the following lemma.
\begin{lemma}
\label{lemma:consumer sovereignty}
The Homo-OMZ mechanism is $O$(1)-competitive for social efficiency.
\end{lemma}

The above five lemmas together prove the following theorem.
\begin{theorem}
The Homo-OMZ mechanism satisfies the computational efficiency, individual rationality, truthfulness, consumer sovereignty, and constant competitiveness for social efficiency under the homogeneous user model and zero arrival-departure interval model.
\end{theorem}

\subsection{Heterogeneous User Mechanism}
Although the \emph{Homo-OMZ} mechanism satisfies four desirable properties (computational efficiency, individual rationality, truthfulness, and consumer sovereignty), and achieves constant competitiveness for social efficiency, the failure of guaranteeing good frugality makes it less attractive.
To achieve good frugality, one possible direction is to make use of the off-the-shelf results on the budgeted feasible mechanisms \cite{singer2013pricing}.
The budgeted feasible mechanism design problem is similar to our frugal mechanism design problem, with the difference that the total payment paid to the winners is a constraint instead of an objective function.
To address this issue, it is intuitive that we can dynamically learn a budget that are enough for allocating users for completing a specific number of tasks, then use this budget to compute a bid threshold by using budget feasible mechanisms, and finally use this bid threshold for making further decisions.
In the following, we present a novel online mechanism \emph{Hetero-OMZ} that satisfies all five desirable properties including the frugality in the heterogeneous user model.
Note that since the homogeneous user model is a special case of the heterogeneous user model, the \emph{Hetero-OMZ} mechanism is a general solution for both two models.

\subsubsection{Mechanism Design}
As illustrated in Algorithm \ref{alg:Hetero-OMZ}, the basic framework of the \emph{Hetero-OMZ} mechanism is very similar to the \emph{Homo-OMZ} mechanism, with two differences: 1) each user may be allocated multiple tasks upper bounded by the number of tasks that it has announced, as long as its bid is not less than the current bid threshold $b^*$, and the allocated stage-task-number $L'$ has not been achieved;
2) it has a different approach to computing the bid threshold, as illustrated in Algorithm \ref{alg:get threshold2}.
\begin{algorithm}
\SetAlgoLined
\caption{Heterogeneous Online Mechanism under Zero Arrival-departure Interval Model (\emph{Hetero-OMZ})}
\label{alg:Hetero-OMZ}
\KwIn{Task number $L$, deadline $T$}
$(t,T',L',\mathcal{S}',b^*,\mathcal{S}) \leftarrow ( 1, \frac{T}{2^{\lfloor \log_2 T \rfloor}}, \frac{L}{2^{\lfloor \log_2 T \rfloor}}, \emptyset, \beta, \emptyset)$\;
\While{$t\leq T$}
{
    \If{there is a user $i$ arriving at time step $t$}
    {
        \eIf{$b_i\leq b^*$ {\bf and} $\sum_{j\in \mathcal{S}}f_j < L'$}
        {
            $f_i\leftarrow \min\{\tau_i,L'-\sum_{j\in \mathcal{S}}f_j\}$; $p_i\leftarrow b^*$\;
            $\mathcal{S}\leftarrow \mathcal{S} \cup \{i\}$\;
        }{
            $f_i\leftarrow 0$; $p_i\leftarrow 0$\;
        }
        $\mathcal{S}'\leftarrow \mathcal{S}' \cup \{i\}$\;
    }
    \If{$t=\lfloor T' \rfloor$}
    {
        $b^*\leftarrow \textbf{GetBidThreshold2}(L',\mathcal{S}')$\;
        $T'\leftarrow 2T'$; $L'\leftarrow 2L'$\;
    }
    $t\leftarrow t+1$\;
}
\end{algorithm}

\begin{algorithm}
\SetAlgoLined
\caption{GetBidThreshold2}
\label{alg:get threshold2}
\KwIn{Stage-task-number $L'$, sample set $\mathcal{S}'$}
$\mathcal{J}\leftarrow \emptyset$; $i \leftarrow \arg\min_{j\in \mathcal{S}'}b_j$\;
\While{$\sum_{j\in \mathcal{J}}f_j < \delta L'$}
{
    $f_i \leftarrow \min\{\tau_i,\delta L'-\sum_{j\in \mathcal{J}}f_j\}$\;
    $\mathcal{J}\leftarrow \mathcal{J}\cup \{i\}$\;
    $i \leftarrow \arg\min_{j\in {\mathcal{S}'\backslash \mathcal{J}}}b_j$\;
}
$B\leftarrow \sum_{j\in \mathcal{J}}f_j b_j$\;
$p\leftarrow \textbf{BudgetFeasibleMechanism}(B,\mathcal{S}')$\;
\Return $p$;
\end{algorithm}

To compute the bid threshold, we first find the minimal cost for performing $\delta L'$ tasks from the sample set $\mathcal{S}'$.
This can be done by using a simple greedy algorithm which sorts users according to their bids, and preferentially allocates tasks to users with lower bids until that all of $\delta L'$ tasks have been allocated (lines 1-6).
Here we set $\delta>1$ to obtain a slight overestimate of the required budget for allocating $L'$ tasks.
Intuitively, if the value of $\delta$ is too small, the budget will be not enough that the required number of tasks can be allocated at the next stage.
Conversely, if the value of $\delta$ is too large, the budget will be wasted, which will result in a bad frugality.
Later we will fix the value of $\delta$ elaborately to enable the mechanism achieving a constant frugality ratio.
Second, we use the total payment for performing $\delta L'$ tasks as the budget $B$ (line 7), and then compute a bid threshold from the sample set $\mathcal{S}'$ and the budget $B$ by using the budget feasible mechanism (line 8).

Note that the budget feasible mechanism is an offline procedure proposed by \cite{singer2013pricing}, which has access all bids.
As illustrated in Algorithm \ref{alg:budgetFeasible}, it also adopts a greedy strategy.
First, users are sorted according to their bids so that: $b_1\leq b_2 \leq \cdots b_{|\mathcal{S'}|}$.
When given a price at $p=b_i$ and the remaining budget $B -p\sum_{j\in \mathcal{J}}f_j$, the number of tasks that a user could be allocated at this price without exceeding the budget constraint is $f_i= \min\{\tau_i,\lfloor \frac{B}{p} \rfloor-\sum_{j\in \mathcal{J}}f_j\}$.
Then we find the largest $k$ such that $b_k\leq \frac{B}{\sum_{j=1}^{k-1}f_j+1}$.
The set of selected users is $\mathcal{J}=\{1,2,\ldots k\}$.
Finally, the bid threshold is set to be $b_k$.
This mechanism follows the \emph{proportional share allocation rule} in essence.
It has a good property that it sets a common price which on one hand is high enough so that enough users could accept and on the other hand is low enough so that the budget could be efficiently exploited.

\begin{algorithm}
\SetAlgoLined
\caption{BudgetFeasibleMechanism \cite{singer2013pricing}}
\label{alg:budgetFeasible}
\KwIn{Budget constraint $B$, sample set $\mathcal{S}'$}
$\mathcal{J}\leftarrow \emptyset$; $i \leftarrow \arg\min_{j\in \mathcal{S}'}b_j$\;
\While{$b_i \leq \frac{B}{\sum_{j\in \mathcal{J}}f_j+1}$}
{
    $p \leftarrow b_i$\;
    $f_i \leftarrow \min\{\tau_i,\lfloor \frac{B}{p} \rfloor-\sum_{j\in \mathcal{J}}f_j\}$\;
    $\mathcal{J}\leftarrow \mathcal{J}\cup \{i\}$\;
    $i \leftarrow \arg\min_{j\in {\mathcal{S}'\backslash \mathcal{J}}}b_j$\;
}
\Return $p$;
\end{algorithm}

In the following, we use an example to illustrate how the \emph{Hetero-OMZ} mechanism works.
\begin{example}
\label{exmaple1}
Consider a crowdsourcer with the number of tasks to perform $L=8$ and the deadline $T=8$.
There are five users arriving online before the deadline with the following profiles: $\theta_1=(1,1,4,2)$, $\theta_2=(2,2,4,4)$, $\theta_3=(4,4,4,5)$, $\theta_4=(6,6,4,1)$, and $\theta_5=(7,7,4,3)$.
\end{example}

We set $\beta=5$ and $\delta=2$. Then the \emph{Hetero-OMZ} mechanism works as follows.
\begin{itemize*}
    \item[$\diamond$] $t=1$: $(T',L',\mathcal{S}',b^*,\mathcal{S})=(1,1,\emptyset,5,\emptyset)$, $f_1=1$, $p_1=5$, $\mathcal{S}=\{1\}$, $\mathcal{S}'=\{1\}$. Update the bid threshold: $b^*=2$.
    \item [$\diamond$] $t=2$: $(T',L',\mathcal{S}',b^*,\mathcal{S})=(2,2,\{1\},2,\{1\})$, $f_2=0$, $p_2=0$, $\mathcal{S}=\{1\}$, $\mathcal{S}'=\{1,2\}$. Update the bid threshold: $b^*=2$.
    \item [$\diamond$] $t=4$: $(T',L',\mathcal{S}',b^*,\mathcal{S})=(4,4,\{1,2\},2,\{1\})$, $f_3=0$, $p_3=0$, $\mathcal{S}=\{1\}$, $\mathcal{S}'=\{1,2,3\}$. Update the bid threshold: $b^*=4$.
    \item [$\diamond$] $t=6$: $(T',L',\mathcal{S}',b^*,\mathcal{S})=(8,8,\{1,2,3\},4,\{1\})$, $f_4=4$, $p_4=4$, $\mathcal{S}=\{1,4\}$, $\mathcal{S}'=\{1,2,3,4\}$.
    \item [$\diamond$] $t\!=\!7\!$: $\!(T',L',\mathcal{S}',b^*,\mathcal{S})\!=\!(8,8,\{1,2,3,4\},4,\{1,4\})$, $f_5=3$, $p_5=4$, $\mathcal{S}=\{1,4,5\}$, $\mathcal{S}'=\{1,2,3,4,5\}$. Finally, the set of selected users is $\mathcal{S}=\{1,4,5\}$, and the payments to these selected 3 users are 5, 16, and 12 respectively.
\end{itemize*}

\subsubsection{Mechanism Analysis}
It can be proved that the \emph{Hetero-OMZ} mechanism also satisfies the computational efficiency, individual rationality, truthfulness, and consumer sovereignty, with very similar proof of the \emph{Homo-OMZ} mechanism.
In the following we prove that the \emph{Hetero-OMZ} mechanism can achieve a constant frugality ratio under both the i.i.d. model and the secretary model by elaborately fixing different value of $\delta$ in Algorithm \ref{alg:get threshold2}.

If the stage-task-number could be achieved at each stage, then $L$ tasks would be allocated finally.
Since our \emph{Hetero-OMZ} mechanism consists of multiple stages, and dynamically increases the stage-task-number, it only needs to prove that $L/2$ tasks could be allocated at the last stage while the total payment is no more than the budget $B$.
In this case, the frugality ratio would be $\delta$, since at the last stage the budget $B$ is the minimal cost for performing $\delta L'=\delta L/2$ tasks according to Algorithm \ref{alg:get threshold2}.
In order to facilitate analysis, we change the stage-task-number constraint into the budget constraint at each stage in Algorithm \ref{alg:Hetero-OMZ}.
Specially, in line 4 of Algorithm \ref{alg:Hetero-OMZ}, the condition $\sum_{j\in \mathcal{S}}f_j < L'$ is replaced with $\sum_{j\in \mathcal{S}}p_j f_j < B$,
and in line 5, the allocation $f_i\leftarrow \min\{\tau_i,L'-\sum_{j\in \mathcal{S}}f_j\}$ is replaced with $f_i\leftarrow \min\{\tau_i,\lfloor(B-\sum_{j\in \mathcal{S}}p_j f_j)/b^*\rfloor\}$.
Note that in this case, if we prove that at least $L/2$ tasks could be allocated at the last stage under the budget constraint $B$, then it is equivalent to that $L/2$ tasks could be allocated while the total payment is no more than $B$, meaning that the frugality ratio is $\delta$.
In the following analysis, we use the modified version of Algorithm \ref{alg:Hetero-OMZ}.

For ease of illustration, we first introduce more notations and concepts.
\begin{itemize*}
\item $\mathcal{Z}$ denotes the set of selected users computed by Algorithm \ref{alg:budgetFeasible} based on the set of users $\mathcal{U}$ arriving before $T$ and the budget $2B$. The bid threshold of $\mathcal{Z}$ is $p$.
\item $\mathcal{S}'$ denotes the sample set obtained when the stage $\lfloor \log_2 T \rfloor$ is over, which consists of all users arriving before the time $\lfloor T/2 \rfloor$.
\item $\mathcal{Z}_1$ and $\mathcal{Z}_2$ denote the subsets of $\mathcal{Z}$ that appears in the first and second half of the input stream, respectively. Thus we have $\mathcal{Z}_1=\mathcal{Z}\cap \mathcal{S}'$, and $\mathcal{Z}_2=\mathcal{Z}\cap \{\mathcal{U}\backslash \mathcal{S}'\}$.
\item $\mathcal{Z}_1'$ denotes the set of selected users computed by Algorithm \ref{alg:budgetFeasible} based on the sample set $\mathcal{S}'$ and the budget $B$.
The bid threshold of $\mathcal{Z}_1'$ is $p_1'$.
\item $\mathcal{Z}_2'$ denotes the set of selected users computed by Algorithm \ref{alg:Hetero-OMZ} at the last stage.
\item The total number of tasks allocated to a set of selected users $\mathcal{X}$ is denoted by a function $f(\mathcal{X})=\sum_{i\in \mathcal{X}}f_i$.
\end{itemize*}

Before analyzing the frugality formally, we introduce a lemma about Algorithm \ref{alg:budgetFeasible}.
\begin{lemma}
\label{lemma:WWW}
 (\cite{singer2013pricing}, Lemma 3.1) For a given sample set of users, let $L$ be the maximal number of tasks that can be allocated under a given budget. Then at least $L/2$ tasks can be allocated under budget at the price computed by Algorithm \ref{alg:budgetFeasible}.
\end{lemma}

According to Lemma \ref{lemma:WWW}, we have $f(\mathcal{Z}_1')\geq \delta L'/2=\delta L/4$.

\begin{lemma}
\label{lemma:iid}
  When $\delta=2$, we have $\mathbb{E}[f(\mathcal{Z}_2')]\geq L/2$ under the i.i.d. model.
\end{lemma}

The proof of Lemma \ref{lemma:iid} is given in Appendix A.

Different from the i.i.d. model, under the secretary model we make $\mathcal{Z}$ denote the set of selected users computed by Algorithm \ref{alg:budgetFeasible} based on the set of users $\mathcal{U}$ arriving before $T$ and the budget $B$.
Other notations and concepts remain unchanged.
In addition, it is assumed that the number of tasks that each user can complete is at most $f(\mathcal{Z})/\omega$.
In order to facilitate analysis, we introduce a lemma as follows.
\begin{lemma}
\label{lemma:secretary_ref}
(\cite{bateni2010submodular}, Lemma 16) For sufficiently large $\omega$, the random variable $|f(\mathcal{Z}_1)-f(\mathcal{Z}_2)|$ is bounded by $f(\mathcal{Z})/2$ with a constant probability.
\end{lemma}

Because $f(\mathcal{X})$ is a linear function, we have $f(\mathcal{Z}_1)+f(\mathcal{Z}_2)=f(\mathcal{Z})$.
Thus, Lemma \ref{lemma:secretary_ref} can be easily extended to the following corollary.
\begin{corollary}
\label{corollary}
For sufficiently large $\omega$, both $f(\mathcal{Z}_1)$ and $f(\mathcal{Z}_2)$ are at least $f(\mathcal{Z})/4$ with a constant probability.
\end{corollary}

\begin{lemma}
\label{lemma:secretary}
  For sufficiently large $\omega$, when $\delta=8$, we have $f(\mathcal{Z}_2')\geq L/2$ with a constant probability under the secretary model.
\end{lemma}

The proof of Lemma \ref{lemma:secretary} is given in Appendix B.
The above analysis proves the following theorem.
\begin{theorem}
The \emph{Hetero-OMZ} mechanism satisfies the computational efficiency, individual rationality, truthfulness, consumer sovereignty, and constant frugality under the zero arrival-departure interval model.
\end{theorem} 

\section{Online Mechanism under General Interval Model}
\label{sec:general case}
In this section, we consider the general interval model, and there may be multiple online users in the auction simultaneously.
Firstly, we change the settings of Example \ref{exmaple1} to show that the \emph{Hetero-OMZ} mechanism is not time-truthful in the general interval model.

\begin{example}
\label{example2}
All the settings are the same as Example \ref{exmaple1} except for that user 1 has a non-zero arrival-departure interval, $a_1<d_1$.
Specially, the profile of user 1 is $\theta_1=(1,5,4,2)$.
\end{example}

In this example, if user 1 report its profile truthfully, then it will obtain the payment 5 according to the \emph{Hetero-OMZ} mechanism.
However, if user 1 delays announcing its arrival time and reports $\theta_1'=(5,5,4,2)$, then it will improve its payment to 20 according to the \emph{Hetero-OMZ} mechanism (the detailed computing process is omitted).

In the following, we will present a new online mechanism, \emph{Hetero-OMG}, and prove that it satisfies all five desirable properties in the general interval model.

\subsection{Mechanism Design}
Since the \emph{Hetero-OMZ} mechanism can be applied in both the homogeneous user model and the heterogeneous user model, and can satisfy several desirable properties, we adopt a similar algorithm framework under the general interval model.
Meanwhile, in order to guarantee the \emph{cost-} and \emph{time-truthfulness}, it is necessary to modify the \emph{Hetero-OMZ} mechanism based on three principles.
First, any user is added to the sample set only when it departs; otherwise, the bid-independence will be destroyed if its arrival-departure time spans multiple stages, because a user can indirectly affect its payment now.
Second, if there are multiple users who have not yet departed at some time, we sort these online users according to the numbers of tasks that they can complete, instead of their bids, and preferentially select those users with higher number of tasks.
In this way, the bid-independence can be held.
Third, whenever a new time step arrives, it scans through the list of users who have not yet departed and allocates tasks to those whose bids are not larger than the current bid threshold under the stage-task-number constraint, even if some arrived much earlier.
At the departure time of any user who was selected as a winner, the user is paid for a payment equal to the maximum payment attained during the user's reported arrival-departure interval, even if this payment is larger than the payment at the time step when the user was selected as a winner.

\begin{algorithm}
\SetAlgoLined
\caption{Heterogeneous Online Mechanism under General Interval Model (\emph{Hetero-OMG})}
\label{alg:Hetero-OMG}
\KwIn{Task Number $L$, deadline $T$}
$(t,T',L',\mathcal{S}',b^*,\mathcal{S}) \leftarrow ( 1, \frac{T}{2^{\lfloor \log_2 T \rfloor}}, \frac{L}{2^{\lfloor \log_2 T \rfloor}}, \emptyset, \beta, \emptyset)$\;
\While{$t\leq T$}
{
    Add all new users arriving at time step $t$ to a set of online users $\mathcal{O}$; $\mathcal{O}'\leftarrow \mathcal{O}\setminus S$\;
    \Repeat{$\mathcal{O}'=\emptyset$}
    {
        $i\leftarrow \arg\max_{j\in \mathcal{O}'}\tau_j$\;
        \eIf{$b_i\leq b^*$ {\bf and} $\sum_{j\in \mathcal{S}}f_j < L'$}
        {
            $f_i\leftarrow \min\{\tau_i,L'-\sum_{j\in \mathcal{S}}f_j\}$; $p_i\leftarrow b^*$\;
            $\mathcal{S}\leftarrow \mathcal{S} \cup \{i\}$\;
        }{
            $f_i\leftarrow 0$; $p_i\leftarrow 0$\;
        }
        $\mathcal{O}'\leftarrow \mathcal{O}'\setminus \{i\}$\;
    }
    Remove all users departing at time step $t$ from $\mathcal{O}$, and add them to $\mathcal{S}'$\;
    \If{$t=\lfloor T' \rfloor$}
    {
        $b^*\leftarrow \textbf{GetBidThreshold2}(L',\mathcal{S}')$\;
        $T'\leftarrow 2T'$; $L'\leftarrow 2L'$; $\mathcal{O}'\leftarrow \mathcal{O}$\;
        \Repeat{$\mathcal{O}'=\emptyset$}
        {
            $i\leftarrow \arg\max_{j\in \mathcal{O}'}\tau_j$\;
            \If{$b_i\leq b^*$ {\bf and} $\min\{\tau_i,L'+f_i-\sum_{j\in \mathcal{S}}f_j\}b^*>f_ip_i$}
            {
                $f_i\leftarrow \min\{\tau_i,L'+f_i-\sum_{j\in \mathcal{S}}f_j\}$; $p_i\leftarrow b^*$\;
                \lIf{$i \notin \mathcal{S}$}
                {
                    $\mathcal{S}\leftarrow \mathcal{S} \cup \{i\}$\;
                }
            }
            $\mathcal{O}'\leftarrow \mathcal{O}'\setminus \{i\}$\;
        }
    }
    $t\leftarrow t+1$\;
}
\end{algorithm}

According to the above principles, we design the \emph{Hetero-OMG} mechanism satisfying all desirable properties under the general interval model, as described in Algorithm \ref{alg:Hetero-OMG}.
Specially, we consider two cases.
The first case is when the current time step $t$ is not at the end of any stage.
In this case, the bid threshold remains unchanged.
The following operations (the lines 3-14 in Algorithm \ref{alg:Hetero-OMG}) are performed.
First, all new users arriving at time step $t$ are added to a set of online users $\mathcal{O}$.
Then we make decision on whether to select these online users one by one in the order of the numbers of tasks that they can complete; the users who can complete more tasks will be selected first.
If an online user $i$ has been selected as a winner before time step $t$, we need not to make decision on it again because it is impossible to obtain a higher payment than before (to be proved later in Appendix C).
Otherwise, we need to make decision on it again: it will be allocated multiple tasks upper bounded by the number of tasks that it has announced, as long as its bid is not less than the current bid threshold $b^*$, and the allocated stage-task-number $L'$ has not been achieved;
meanwhile, we pay user $i$ at price $p_i=b^*$ per task, and add it to the set of selected users $\mathcal{S}$.
Finally, we remove all users departing at time step $t$ from $\mathcal{O}$, and add them to the sample set $\mathcal{S}'$.

The second case is when the current time step is just at the end of some stage.
In this case, the bid threshold will be updated.
The mechanism works as the lines 16-25.
We need to make decision on whether to select these online users, and at what prices, one by one in the order of the numbers of tasks that they can complete, no matter whether they have ever been selected as the winners before time step $t$.
As shown in the lines 20-23, if user $i$ can obtain a higher payment than before (this could be because that this user is given a higher price or is allocated more tasks), its price or number of tasks allocated will be updated.
Meanwhile, if user $i$ has never been selected as a winner before time step $t$, it will be added to the set $\mathcal{S}$.

Return to Example \ref{example2}.
If all of the five users report their types truthfully, then the \emph{Hetero-OMG} mechanism works as follows.
\begin{itemize*}
    \item[$\diamond$] $t=1$: $(T',B',\mathcal{S}',b^*,\mathcal{S})=(1,1,\emptyset,5,\emptyset)$, $f_1=1$, $p_1=5$, $\mathcal{S}=\{1\}$, $\mathcal{S}'=\emptyset$. Update the bid threshold: $b^*=5$, update the allocation and payment of user 1: $f_1=2$, $p_1=5$.
    \item [$\diamond$] $t=2$: $(T',B',\mathcal{S}',b^*,\mathcal{S})=(2,2,\emptyset,5,\{1\})$, $f_2=0$, $p_2=0$, $\mathcal{S}=\{1\}$, $\mathcal{S}'=\{2\}$. Update the bid threshold: $b^*=4$, update the allocation and payment of user 1: $f_1=4$, $p_1=4$.
    \item [$\diamond$] $t=4$: $(T',B',\mathcal{S}',b^*,\mathcal{S})=(4,4,\{2\},4,\{1\})$, $f_3=0$, $p_3=0$, $\mathcal{S}=\{1\}$, $\mathcal{S}'=\{2,3\}$. Update the bid threshold: $b^*=5$, update the allocation and payment of user 1: $f_1=4$, $p_1=5$.
    \item [$\diamond$] $t=5$: user 1 departs, so $\mathcal{S}'=\{1,2,3\}$.
    \item [$\diamond$] $t=6$: $(T',B',\mathcal{S}',b^*,\mathcal{S})=(8,8,\{1,2,3\},5,\{1\})$, $f_4=4$, $p_4=5$, $\mathcal{S}=\{1,4\}$, $\mathcal{S}'=\{1,2,3,4\}$. Now all of 8 tasks have been allocated.
    \item [$\diamond$] $t\!=\!7\!$: $\!(T',B',\mathcal{S}',b^*,\mathcal{S})\!=\!(8,8,\{1,2,3,4\},5,\{1,4\})$, $f_5=0$, $p_5=0$, $\mathcal{S}=\{1,4\}$, , $\mathcal{S}'=\{1,2,3,4,5\}$.
\end{itemize*}

Thus, user 1 can obtain the payment 20 according to the \emph{Hetero-OMG} mechanism.
Even if user 1 delays announcing its arrival time and reports $\theta_1'=(5,5,4,2)$, it still cannot improve its payment (the detailed computing process is omitted).
Therefore, the time-truthfulness can be guaranteed in the general interval model.

\subsection{Mechanism Analysis}
It is convenient to prove that the \emph{Hetero-OMG} mechanism also satisfies the \emph{computational efficiency}, \emph{individual rationality}, \emph{consumer sovereignty}, and \emph{constant frugality} as \emph{Hetero-OMZ} (with almost the same proof), although \emph{Hetero-OMG} may have slightly higher frugality ratio than \emph{Hetero-OMZ}.
Most importantly, we can prove the truthfulness of the \emph{Hetero-OMG} mechanism.
\begin{lemma}
\label{lemma:time-truthful}
  The \emph{Hetero-OMG} mechanism is cost- and time-truthful.
\end{lemma}
The proof of Lemma \ref{lemma:time-truthful} is given in Appendix C.

\begin{theorem}
The \emph{Hetero-OMG} mechanism satisfies the computational efficiency, individual rationality, budget feasibility, truthfulness, consumer sovereignty, and constant frugality under the general interval model.
\end{theorem} 

\section{Performance Evaluation}
\label{sec:performance evaluation}
To evaluate the performance of our online mechanisms, we implemented \emph{Homo-OMZ}, \emph{Hetero-OMZ} and \emph{Hetero-OMG}, and compared them against the following two benchmarks.
The first benchmark is the \emph{optimal} offline solution which has full knowledge about all users' profiles.
This can be done by using a simple greedy algorithm as illustrated in Algorithm \ref{alg:get threshold2} (lines 1-6), which sorts users according to their bids, and preferentially allocates tasks to users with lower bids until that all of tasks have been allocated.
The second benchmark is the \emph{random} mechanism, which adopts a naive strategy, i.e., rewards users based on an uninformed fixed bid threshold.

\subsection{Evaluation under Homogeneous User Model and Zero Arrival-departure Interval Model}
\label{subsec:homo evaluation}
We first evaluate the \emph{Homo-OMZ} and \emph{Hetero-OMZ} mechanisms under the homogeneous user model and zero arrival-departure interval model.

\underline{Simulation Setup}:
We set the deadline to $T$=1800s, and vary the number $L$ of tasks to be completed from 100 to 400 with the increment of 100.
Users arrive according to a Poisson process in time with arrival rate $\lambda=0.6$.
For each user $i \in \mathcal{U}$, we set that: $a_i=d_i$, $\tau_i=1$, $c_i\sim U[1,10]$, where $U[a,b]$ denotes a uniform distribution over $[a,b]$.
For the \emph{Hetero-OMZ} mechanism, we set $\delta=1$ and $\delta=2$ respectively, and $\beta=10$.
As we proved in Lemma \ref{lemma:iid}, when $\delta=2$ the \emph{Hetero-OMZ} mechanism is constant-\emph{frugal}.
Here we also set $\delta=1$ for comparison.
For the \emph{optimal} offline mechanism, we compute the minimum cost of performing $L$ and $2L$ tasks in order to obtain frugality ratios under the two frameworks of ``\emph{idealistic frugality}" and ``\emph{realistic frugality}", respectively.
For the \emph{random} mechanism, we obtain the average performance of 50 such solutions for evaluations, where in each solution the bid threshold is chosen at random from the range of 1 to 10.
\begin{figure}[!t]
  \centering{
  \subfigure[The total payment]{
    \label{fig_homoPayment}
    \includegraphics[height=1.8in]{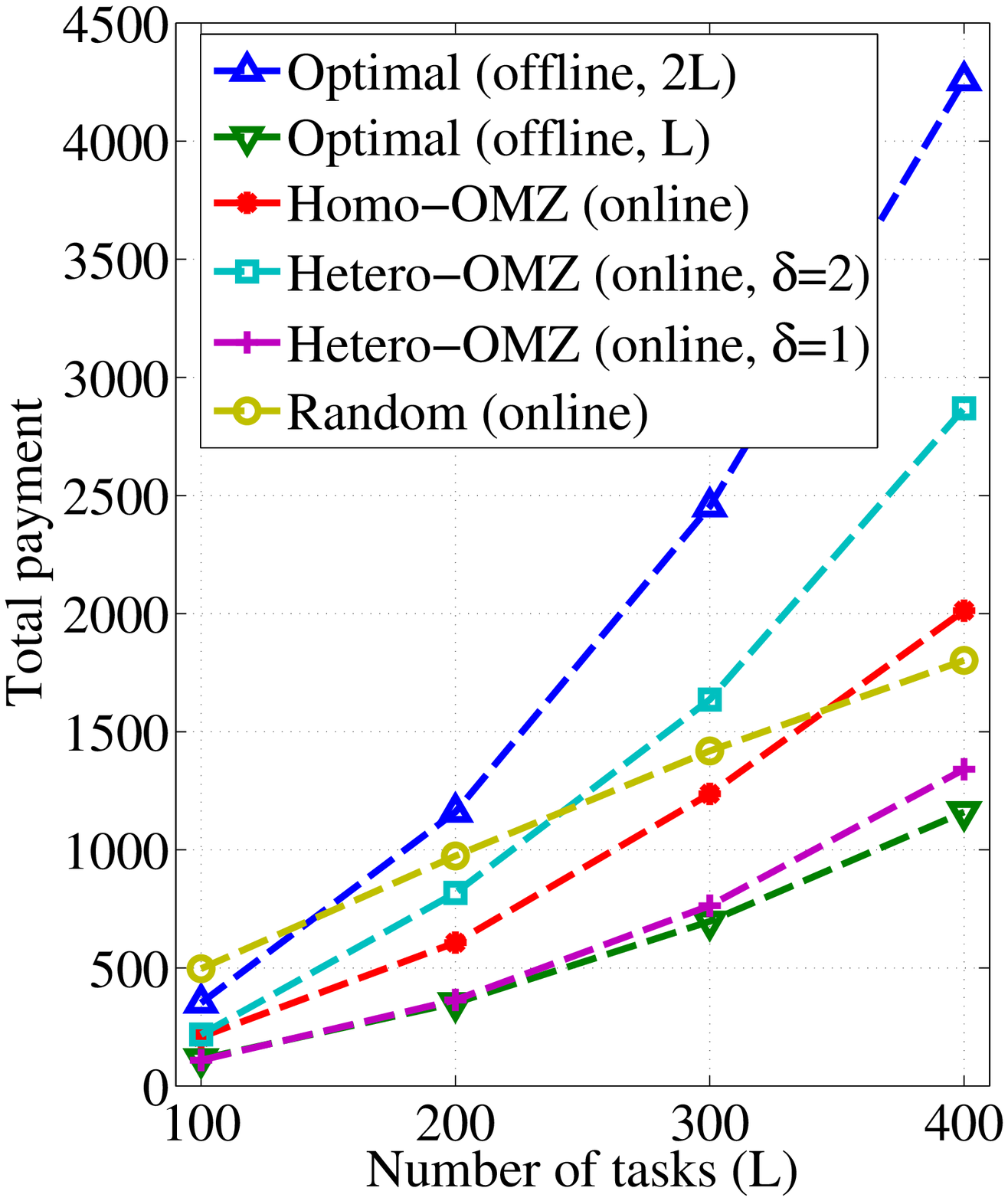}}
  \subfigure[The number of completed tasks]{
    \label{fig_homoTask}
    \includegraphics[height=1.8in]{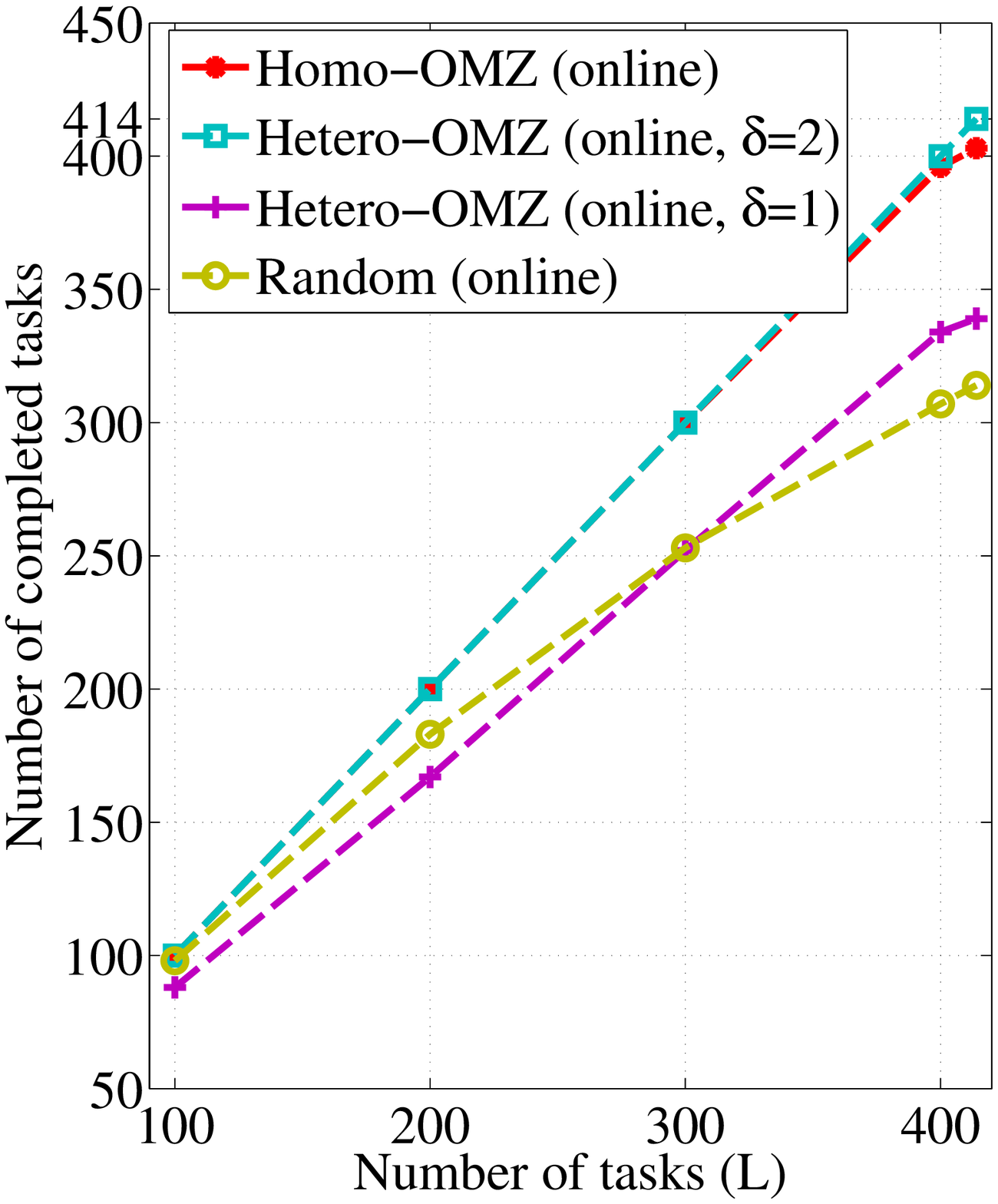}}
  }
  \caption{Evaluation results under the homogeneous user model and zero arrival-departure interval model.}
  \label{fig_HomoResults} 
\end{figure}

\underline{Simulation Results}:
Fig. \ref{fig_homoPayment} compares the crowdsourcer's total payment required by the \emph{Homo-OMZ} and \emph{Hetero-OMZ} mechanisms against the two benchmarks.
Fig. \ref{fig_homoTask} compares the number of tasks completed by the three online mechanisms.
From the simulation results, we can observe the following four phenomena:
\begin{itemize*}
\item The total payments of all evaluated mechanisms increase with the number of tasks, and the growth rate is larger than 1 (except for the \emph{random} mechanism).
It is because that the user set is limited, and the users with higher costs must be selected if more tasks are required to be completed.
\item Both of the \emph{Homo-OMZ} and \emph{Hetero-OMZ} mechanisms outweigh the \emph{random} mechanism.
Note that, although the \emph{random} mechanism has less total payment under some cases, it has many tasks uncompleted, and the average price per task is higher than that of the other mechanisms.
\item Only the \emph{Hetero-OMZ} mechanism with $\delta=2$ can complete all of required tasks under various cases, although it requires higher total payment than the \emph{Homo-OMZ} mechanism and the \emph{Hetero-OMZ} mechanism with $\delta=1$. Note that, at most 828 tasks can be completed (by the optimal offline mechanism) in our simulations, and a half of these tasks (414 tasks) can be completed by the \emph{Hetero-OMZ} mechanism with $\delta=2$ (see Fig. \ref{fig_homoTask}).
    In fact, it implies a trade-off between the frugality and the number of completed tasks, which should be considered by the mechanism designer.
    Anyway, the \emph{Hetero-OMZ} mechanism is a good choice to guarantee the completion of required tasks without sacrificing the frugality.
\item The total payment of the \emph{Hetero-OMZ} mechanism with $\delta=2$ is less than the \emph{optimal} offline mechanism with $2L$ tasks, indicating that the ``realistic" frugality ratio is less than 2, which is consistent with our theoretical analysis in Lemma \ref{lemma:iid}.
    Moreover, by comparing the total payments of the \emph{Hetero-OMZ} mechanism with $\delta=2$ and the \emph{optimal} offline mechanism with $L$ tasks, we have that the ``idealistic" frugality ratio is less than 2.5. Anyway, it has been verified that the \emph{Hetero-OMZ} mechanism can guarantee the constant frugality.
\end{itemize*}

\subsection{Evaluation under Heterogeneous User Model}
Now we evaluate the \emph{Hetero-OMZ} and \emph{Hetero-OMG} mechanisms under the heterogeneous user model.

\underline{Simulation Setup}:
We set the deadline to $T$=1800s, and vary $L$ from 100 to 1000 with the increment of 100.
Users arrive according to a Poisson process in time with arrival rate $\lambda$, and we vary $\lambda$ from 0.2 to 1 with the increment of 0.2.
For each user $i \in \mathcal{U}$, we set that: $\tau_i\sim U[1,10]$, $c_i\sim U[1,10]$.
For the \emph{Hetero-OMZ} mechanism, each user has zero arrival-departure interval.
For the \emph{Hetero-OMG} mechanism, the arrival-departure interval of each user is uniformly distributed over [0,300] seconds.
For both of the above two mechanisms, we set $\delta=2$ and $\beta=10$.
The settings of the \emph{optimal} offline mechanism and the \emph{random} mechanism are the same as in Section \ref{subsec:homo evaluation}.

\underline{Simulation Results}:
Fig. \ref{fig_HeteroResults1} compares the crowdsourcer's total payment required by the \emph{Hetero-OMZ} and \emph{Hetero-OMG} mechanisms against the two benchmarks .
Fig. \ref{fig_HeteroResults2} compares the average price per task required by various mechanisms.
From the simulation results, we can observe the following four phenomena:
\begin{itemize*}
\item From Fig. \ref{fig_payment_rate} and Fig. \ref{fig_price_rate}, we can observe that all evaluated mechanisms (expect for the \emph{random} mechanism) require less total payment and average price per task when more users participate.
    From Fig. \ref{fig_payment_task} and Fig. \ref{fig_price_task}, we can observe that all evaluated mechanisms require higher total payment and average price per tasks (except for the \emph{random} mechanism) when the number of tasks increases.
\item Both of the \emph{Hetero-OMZ} and \emph{Hetero-OMG} mechanisms outweigh the \emph{random} mechanism in terms of both the total payment and average price per task.
Note that, although the \emph{random} mechanism has less total payment or average price per task under some cases, it has many tasks uncompleted.
By contrast, both of the other mechanisms can complete all required tasks.
\item The total payment of the \emph{Hetero-OMZ} mechanism is less than the \emph{optimal} offline mechanism with $2L$ tasks, indicating that the ``realistic" frugality ratio is less than 2, which is consistent with our theoretical analysis in Lemma \ref{lemma:iid}.
    Moreover, by comparing the total payments of the \emph{Hetero-OMZ} mechanism and the \emph{optimal} offline mechanism with $L$ tasks, we have that the ``idealistic" frugality ratio is less than 2.4. Anyway, it has been verified that the \emph{Hetero-OMZ} mechanism can guarantee the constant frugality.
\item The total payment of the \emph{Hetero-OMG} mechanism is very close to the \emph{optimal} offline mechanism with $2L$ tasks, indicating that the ``realistic" frugality ratio is close to 2.
    Besides, the \emph{Hetero-OMG} mechanism requires higher total payment and average price per tasks than the \emph{Hetero-OMZ} mechanism in order to guarantee the time-truthfulness.
\end{itemize*}

\begin{figure}[!t]
  \centering{
  \subfigure[Impact of $\lambda$ ($L=500$)]{
    \label{fig_payment_rate}
    \includegraphics[height=1.9in]{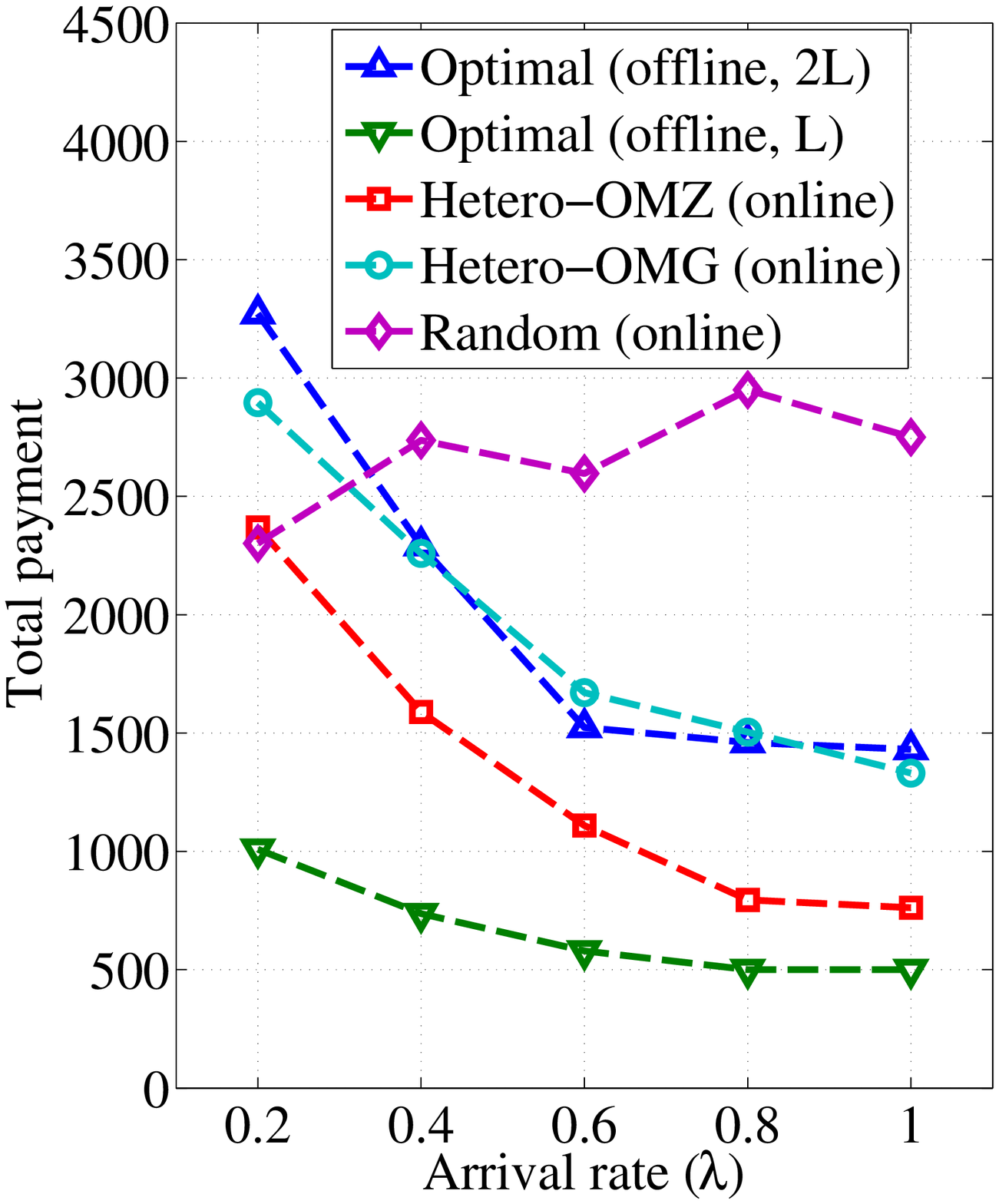}}
  \subfigure[Impact of $L$ ($\lambda=0.6$)]{
    \label{fig_payment_task}
    \includegraphics[height=1.9in]{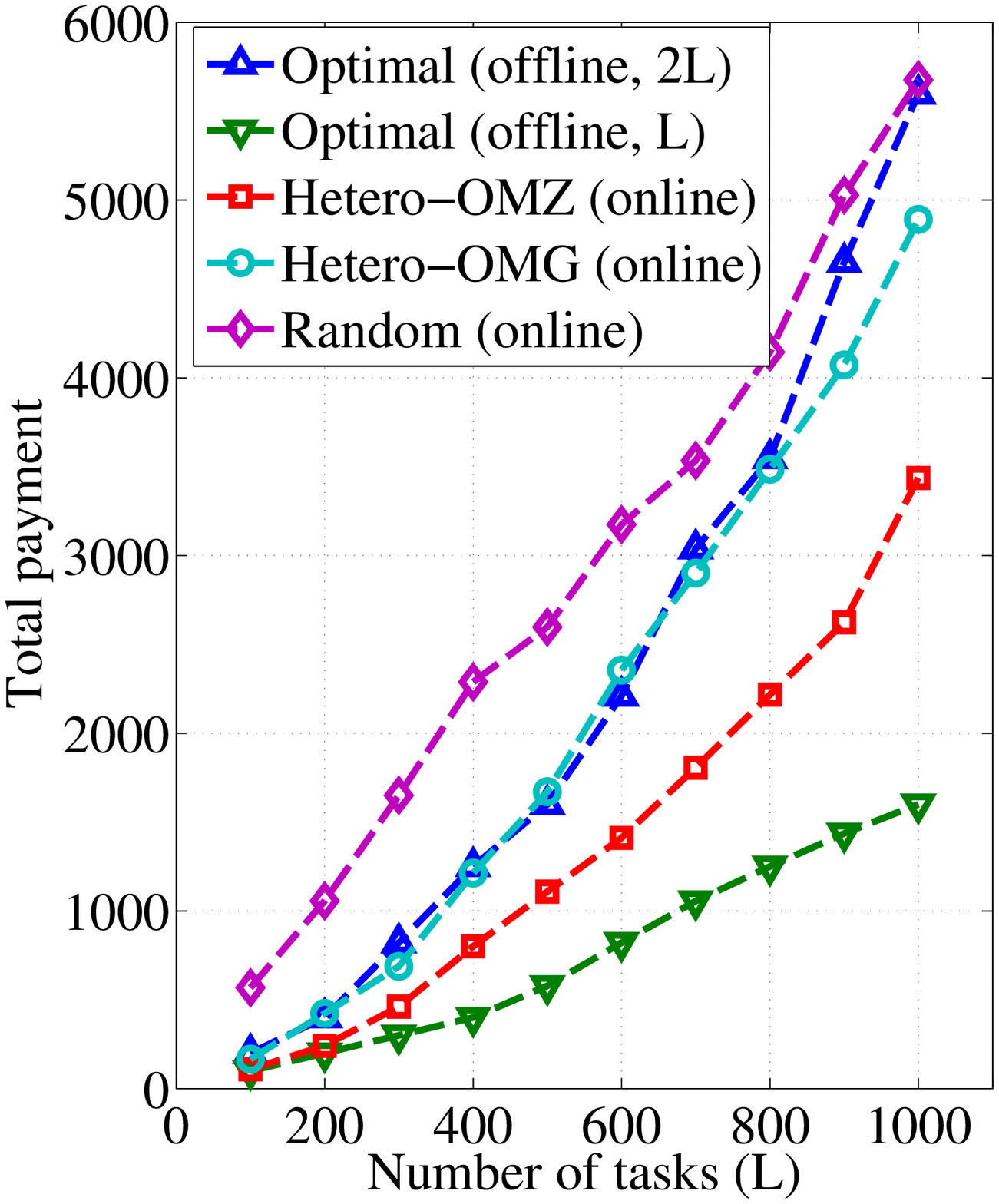}}
  }
  \caption{The total payment under the heterogeneous user model.}
  \label{fig_HeteroResults1} 
\end{figure}

\begin{figure}[!t]
  \centering{
  \subfigure[Impact of $\lambda$ ($L=500$)]{
    \label{fig_price_rate}
    \includegraphics[height=1.9in]{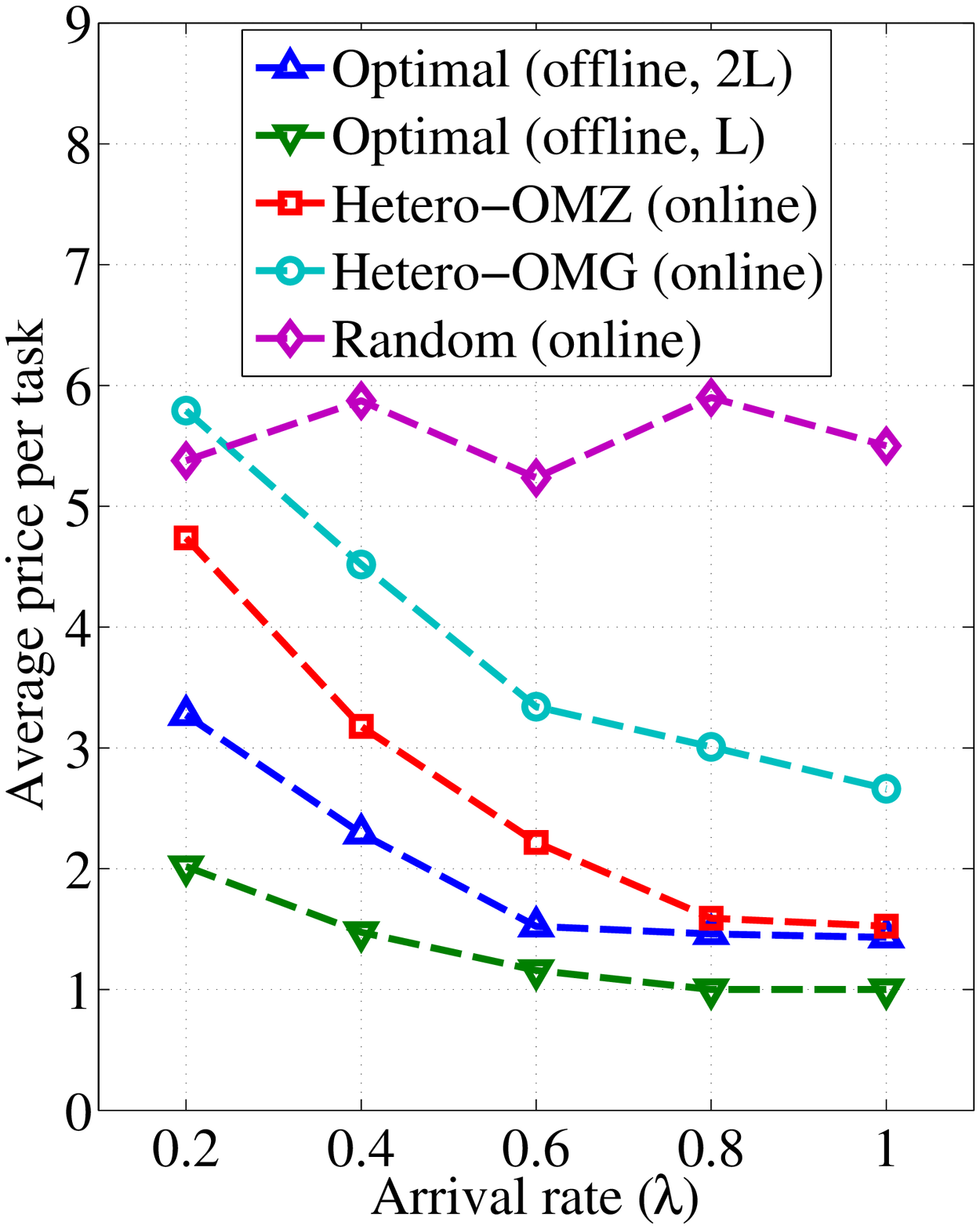}}
  \subfigure[Impact of $L$ ($\lambda=0.6$)]{
    \label{fig_price_task}
    \includegraphics[height=1.9in]{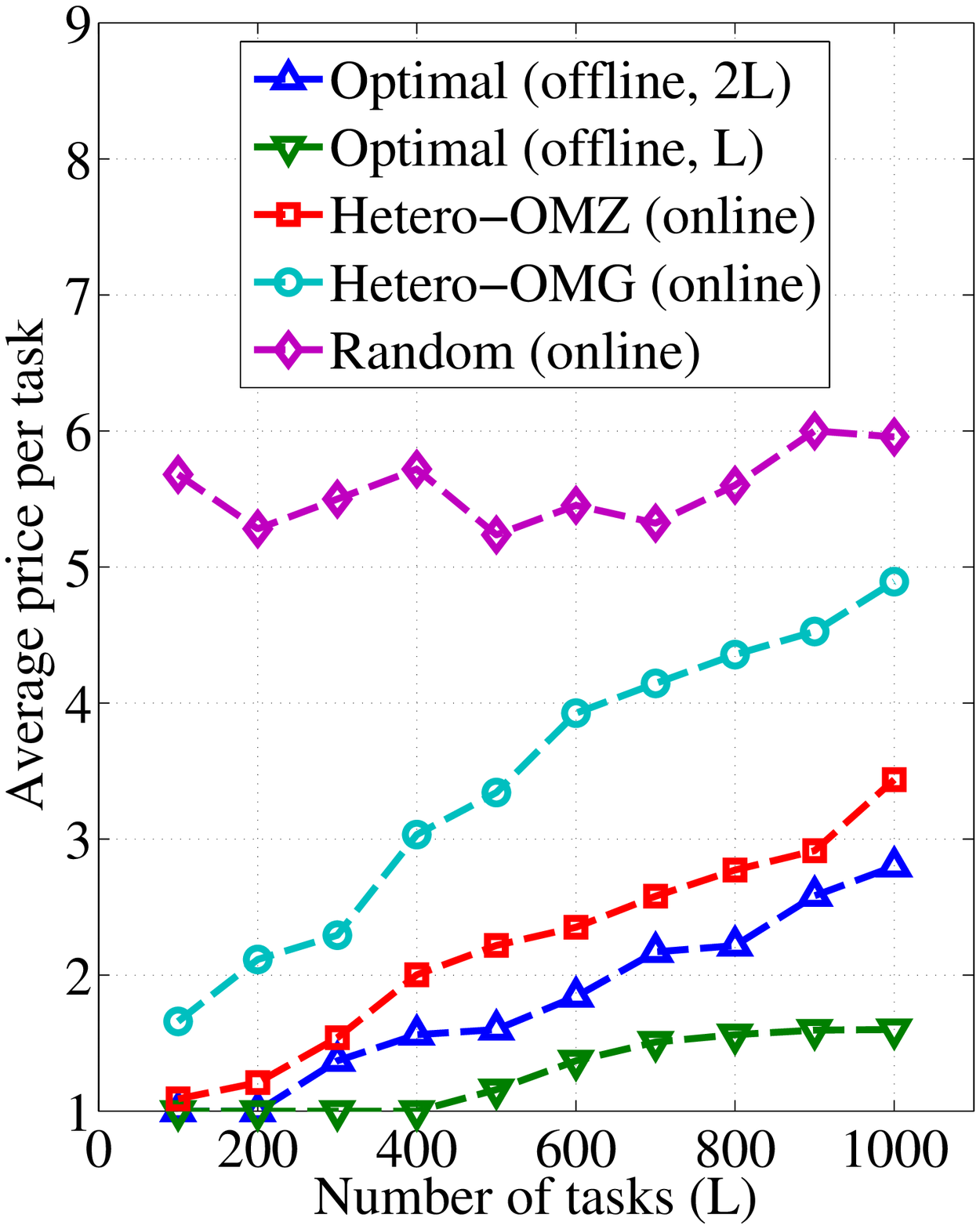}}
  }
  \caption{The average price per task under the heterogeneous user model.}
  \label{fig_HeteroResults2} 
\end{figure}

\section{Related Work}
\label{sec:related work}
\subsection{Mechanism Design for Mobile Crowd Sensing}
Reddy et al. \cite{reddy2010recruitment} developed recruitment frameworks to enable the crowdsourcer to identify well-suited participants for data collections.
However, they focused only on the user selection instead of the incentive mechanism design.
At present, there are some studies \cite{danezis2005much,lee2010sell,duan2012incentive,yang2012crowdsourcing,jaimes2012location} on incentive mechanism design for MCS applications in the offline scenario.
Generally, two system models are considered: the platform/crowdsourcer-centric model where the crowdsourcer provides a fixed reward to participating users, and the user-centric model where users can have their reserve prices for the sensing service.
For the crowdsourcer-centric model, incentive mechanisms were designed by using a Stackelberg game \cite{yang2012crowdsourcing,duan2012incentive}.
The Nash Equilibrium and Stackelberg Equilibrium were computed as the solution, where the costs of all users or their probability distribution was assumed to be known.
In contrast, the user-centric model allows that each user has a private cost only known to itself.
Danezis et al. \cite{danezis2005much} developed a sealed-bid second-price auction to estimate the users' value of sensing data with location privacy.
Lee and Hoh \cite{lee2010sell} designed and evaluated a reverse auction based dynamic price incentive mechanism, where users can sell their sensed data to a service provider with users' claimed bids.
Jaimes et al. \cite{jaimes2012location} proposed a recurrent reverse auction incentive mechanism with a greedy algorithm that selects a representative subset of the users according to their location given a fixed budget.
Yang et al. \cite{yang2012crowdsourcing} designed an auction-based incentive mechanism, and proved this mechanism was computationally efficient, individually rational, profitable, and truthful.
However, all of these studies failed to account for the online arrival of users.

Recently, some researchers have begun to focus on online mechanism design for crowdsourcing markets \cite{singer2013pricing,singla2013truthful,badanidiyuru2012learning,zhao2014crowdsourcing,zhang2014free}.
Singer et al. \cite{singer2013pricing} and Singla et al. \cite{singla2013truthful} presented incentive mechanisms for maximizing a linear utility function based on the bidding model and the posted price model respectively.
Badanidiyuru et al. \cite{badanidiyuru2012learning} and Zhao et al. \cite{zhao2014crowdsourcing} considered incentive mechanisms for maximizing a submodular utility function.
However, all of the above studies mainly focus on budget feasible mechanisms.
Zhang et al. \cite{zhang2014free} considered incentive mechanisms for maximizing the platform utility, which is defined as the total value of selected users minus the total cost.
Only Singer et al. \cite{singer2013pricing} considered the frugal mechanisms partially, but it failed to consider the \emph{consumer sovereignty} and \emph{time-truthfulness}.

\subsection{Online Auctions and Generalized Secretary Problems}
\emph{Online auction} is the essence of many networked markets, in which information about goods, agents, and outcomes is revealed one by one online in a random order, and the agents must make irrevocable decisions without knowing future information.
Combining optimal stopping theory with game theory provides us a powerful tool to model the actions of rational agents in an online auction.
The theory of \emph{optimal stopping} is concerned with the problem of choosing a time to take a particular action, in order to maximize an expected reward or minimize an expected cost.
A classic problem of optimal stopping theory is the \emph{secretary problem}: designing an algorithm for hiring one secretary from a pool of $n$ applicants arriving online, to maximize the probability of hiring the best secretary \cite{dynkin1963optimum}.
Many variants of the classic secretary problem have been studied in the literature and the most relevant to this work is the \emph{k-choice secretary problem}, in which the interviewer is allowed to hire up to $k \geq 1$ applicants in order to maximize performance of the secretarial group based on their overlapping skills.
Kleinberg \cite{kleinberg2005multiple} and Babaioff et al. \cite{babaioff2007knapsack} presented two constant competitive algorithms for a special $k$-choice secretary problem in which the objective function is a linear one, equaling to the sum of the individual values of selected applicants.
They could also be used for minimizing the total cost of selected applicants in essence.
However, they failed to consider the \emph{truthfulness} and \emph{consumer sovereignty}.
Although some solutions (\cite{hajiaghayi2004adaptive,babaioff2008online,kleinberg2005multiple}) of online auctions provided good ideas of designing truthful mechanisms, they still could not satisfy the \emph{consumer sovereignty}.
Moreover, all of these solutions only applied to the homogeneous user model instead of the heterogeneous user model, and could only satisfy the social efficiency instead of the frugality.

\section{Conclusions and Future Work}
\label{sec:conclusion}
In this paper, we have investigated online incentive mechanisms for mobile crowd sensing.
We focus on frugal mechanisms which aim at minimizing the total payment while a specific number of tasks can be completed.
We have designed three online mechanisms that are applicable to different models and hold different properties.
The \emph{Homo-OMZ} mechanism is applicable to the homogeneous user model and can satisfy the \emph{social efficiency} but not \emph{constant frugality}.
The \emph{Hetero-OMZ} and \emph{Hetero-OMG} mechanisms are applicable to both the homogeneous and heterogeneous user models, and can satisfy the \emph{constant frugality}.
The \emph{Hetero-OMG} mechanism can also satisfy the \emph{time-truthfulness} under a general interval model.
Besides, all of these three mechanisms can satisfy the \emph{computational efficiency}, \emph{individual rationality}, \emph{cost-truthfulness}, and \emph{consumer sovereignty}.

An interesting open problem is to design frugal online incentive mechanisms for more complex scenarios.
For example, each user can complete a subset of tasks, and the crowdsourcer wants to complete the whole set of tasks, or obtain a specific value from selected users where the value function is submodular.


\bibliographystyle{IEEEtran}
\bibliography{myRef}
\section*{APPENDIX}
\subsection{Proof of Lemma \ref{lemma:iid}}
Since the numbers of tasks that users can complete and the respective unit costs are i.i.d., each user can be selected in the set $\mathcal{Z}$ with the same probability. The sample set $\mathcal{S}'$ is a random subset of $\mathcal{U}$ since all users arrive in a random order.
Therefore, the number of users from $\mathcal{Z}$ in the sample set $\mathcal{S}'$ follows a hypergeometric distribution $H(n/2,|\mathcal{Z}|,n)$.
Thus, we have $\mathbb{E}[|\mathcal{Z}_1|]=\mathbb{E}[|\mathcal{Z}_2|]=|\mathcal{Z}|/2$.
The number of tasks allocated to each user can be seen as an i.i.d. random variable, and because $f(\mathcal{X})$ is a linear function, it can be derived that: $\mathbb{E}[f(\mathcal{Z}_1)]=\mathbb{E}[f(\mathcal{Z}_2)]= \mathbb{E}[f(\mathcal{Z})]/2$.
The expected total payments to the users from both $\mathcal{Z}_1$ and $\mathcal{Z}_2$ are $B$.
Since $f(\mathcal{Z}_1')$ is computed under budget $B$, it can be derived that: $\mathbb{E}[f(\mathcal{Z}_1')]= \mathbb{E}[f(\mathcal{Z}_1)]= \mathbb{E}[f(\mathcal{Z})]/2$, and $\mathbb{E}[p_1']=\mathbb{E}[p]$.

We consider two cases according to the total payment to the selected users at the last stage as follows.
\begin{list}{}{\setlength{\leftmargin}{0cm}}
\item \textbf{Case 1)}: The budget $B$ is exhausted.

In this case, since the price per task paid to each selected user is $p_1'$, so we have that
\[f(\mathcal{Z}_2') = \frac{B}{p_1'} \geq f(\mathcal{Z}_1')\geq \frac{\delta L}{4} \geq \frac{L}{2},\]
where the first inequality follows from the fact that the total payment to $\mathcal{Z}_1'$ should be no more than the budget $B$, i.e., $p_1'f(\mathcal{Z}_1')\leq B$.
\item \textbf{Case 2)}: The budget $B$ is not exhausted.

Since $\mathbb{E}[p_1']=\mathbb{E}[p]$, each user $i\in \mathcal{Z}_2$ will be allocated in $\mathcal{Z}_2'$.
It implies that:
\[\mathbb{E}[f(\mathcal{Z}_2')] = \mathbb{E}[f(\mathcal{Z}_2)] \geq \frac{L}{2}.\]
\end{list}

Combining case 1) and case 2), we have that at least $L/2$ tasks will be allocated at the last stage in expectation.

\subsection{Proof of Lemma \ref{lemma:secretary}}
  We consider two cases according to the total payment to the selected users at the last stage as follows.
\begin{list}{}{\setlength{\leftmargin}{0cm}}
\item \textbf{Case 1)}: The budget $B$ is exhausted.

In this case, since the price per task paid to each selected user is $p_1'$, so we have that
\[f(\mathcal{Z}_2') = \frac{B}{p_1'} \geq f(\mathcal{Z}_1')\geq \frac{\delta L}{4} \geq 2L,\]
where the first inequality follows from the fact that the total payment to $\mathcal{Z}_1'$ should be no more than the budget $B$, i.e., $p_1'f(\mathcal{Z}_1')\leq B$.
\item \textbf{Case 2)}: The budget $B$ is not exhausted.

Because $p_1'$ is computed over a smaller subset compared with $p$, we have $p_1'\geq p$.
Therefore, for each user $i\in \mathcal{Z}_2$ it follows that $b_i\leq p \leq p_1'$, and all users in $\mathcal{Z}_2$ will be allocated.
It implies that:
\[f(\mathcal{Z}_2') \geq f(\mathcal{Z}_2) \geq \frac{f(\mathcal{Z})}{4} \geq \frac{f(\mathcal{Z}_1')}{4} \geq \frac{L}{2},\]
where the second inequality follows from Corollary \ref{corollary}, and the third inequality follows from the fact that $f(\mathcal{Z})$ is computed over the whole user set $\mathcal{U}$, while $f(\mathcal{Z}_1')$ is computed over a smaller sample set $\mathcal{S}'$ with the same budget $B$.
\end{list}

Combining case 1) and case 2), we have that at least $L/2$ tasks will be allocated at the last stage with a constant probability.

\subsection{Proof of Lemma \ref{lemma:time-truthful}}
Consider a user $i$ with true type $\theta_i=(a_i,d_i,\tau_i,c_i)$, and reported strategy type $\hat{\theta_i}=(\hat{a_i},\hat{d_i},\tau_i,b_i)$.
According to the \emph{Hetero-OMG} mechanism, at each time step $t\in [\hat{a_i},\hat{d_i}]$, there may be a new decision on how many tasks allocated to user $i$, and at what price per task.
For convenience, let $T'_t$, $L'_t$, $b^*_t$, and $\mathcal{S}_t$ denote the end time of the current stage, the residual number of tasks, the current bid threshold, and the set of selected users respectively at time step $t$ and before making decision on user $i$.
Let $\hat{\theta}_{-i}$ denote the strategy types of all users excluding $\hat{\theta_i}$.
We first prove the following two propositions.

\emph{Proposition (a): at some time step $t\in [\hat{a_i},\hat{d_i}]$, fix $b^*_t$ and $L'_t$, reporting the true cost is a dominant strategy for user $i$.}
It can be easily proved since the decision at time step $t$ is bid-independent.

\emph{Proposition (b): fix $b_i$ and $\hat{\theta}_{-i}$, reporting the true arrival/departure time is a dominant strategy for user $i$.}
It's because that user $i$ is always given a payment equal to the maximum payment attained during its reported arrival-departure interval.
Assume that user $i$ can obtain the maximum payment at time step $t\in [\hat{a_i},\hat{d_i}]$.
Then reporting an earlier arrival time or a later departure time than $t$ does not affect the payment of user $i$.
However, if user $i$ reports a later arrival time or an earlier departure time than $t$, then it will obtain a lower payment.

Based on the proposition (b), it is sufficient to prove this lemma by adding a third proposition:

\emph{Proposition (c): fix $[a_i,d_i]$ and $\hat{\theta}_{-i}$, reporting the true cost is a dominant strategy for user $i$.}
According to the proposition (a), reporting a false cost at time step $t$ cannot improve user $i$'s payment at the current time.
Thus, it only needs to prove that \emph{reporting a false cost at time step $t \in [a_i,d_i)$ still cannot improve user $i$'s payment at time step $t' (t<t'\leq d_i)$}.
In the following, we consider two cases according to whether user $i$ is selected as a winner by reporting its true type at time step $t=a_i$.

\begin{list}{}{\setlength{\leftmargin}{0cm}}
\item \textbf{Case 1)}: User $i$ is a winner at time step $t=a_i$.

In this case it satisfies $c_i\leq b^*_t$ and $\sum_{j\in \mathcal{S}_t}f_j < L'_t$, and it can obtain the payment $\min\{\tau_i,L'_t-\sum_{j\in \mathcal{S}_t}f_j\} b^*_t$.
At time $t' (t<t'<T'_t)$, user $i$ will obtain the payment $\min\{\tau_i,L'_{t'}-\sum_{j\in \mathcal{S}_{t'}}f_j\} b^*_t$ if it satisfies $b_i\leq b^*_t$ and $\sum_{j\in \mathcal{S}_{t'}}f_j < L'_{t'}$, otherwise it will obtain the payment 0.
Because $L'_{t'}-\sum_{j\in \mathcal{S}_{t'}}f_j \leq L'_t-\sum_{j\in \mathcal{S}_t}f_j$, user $i$ cannot obtain higher payment at time step $t'$ than that at $t$.
It implies that a user cannot improve its payment by reporting a false cost if its arrival-departure interval does not span more than one stage.

Now we consider user $i$'s payment at time step $t'(T'_t \leq t' \leq d_i)$ if its arrival-departure interval spans multiple stages.
According to the proposition (a), user $i$'s payment at time step $t'$ depends on $b^*_{t'}$ and $L'_{t'}$.
Because $b^*_{t'}$ is independent with $b_i$, it only needs to consider the effect of $b_i$ on $L'_{t'}$.
If user $i$ reports a false cost $b_i$ that still satisfies $b_i\leq b^*_t$ and $\sum_{j\in \mathcal{S}_{t}}f_j < L'_{t}$, then it is still accepted with payment $\min\{\tau_i,L'_t-\sum_{j\in \mathcal{S}_t}f_j\} b^*_t$ at time step $t$, and thus $L'_{t'}$ remains unchanged.
If user $i$ reports a larger bid $b_i>c_i$ and $b_i > b^*_t$, then it will not selected at time step $t$.
In this case, more tasks will be allocated to other users, and $L'_{t'}$ will be diminished.
Therefore, user $i$ cannot obtain higher payment at time step $t'$.

\item \textbf{Case 2)}: User $i$ is not a winner at time step $t=a_i$.

In this case it satisfies $c_i>b^*_t$, or $\sum_{j\in \mathcal{S}_t}f_j = L'_t$.
In case $c_i>b^*_t$, if user $i$ reports a false cost $b_i$ which still satisfies $b_i>b^*_t$, then the outcome remains unchanged.
If user $i$ reports a lower bid $b_i<c_i$ and $b_i\leq b^*_t$, then it will be accepted at price $b^*_t$ at time step $t$.
In such case, however, its utility will be negative.
In addition, $L'_{t'}$ remains unchanged, and thus user $i$'s payment at time step $t'>t$ is not affected.
In case $\sum_{j\in \mathcal{S}_t}f_j = L'_t$, reporting a false cost does not affect the outcome at time step $t$ or the residual budget $L'_{t'}$ at time step $t'>t$.
To sum up, reporting a false cost cannot improve user $i$'s payment at time step $t'>t$.
\end{list}

Therefore, the above three propositions together complete the proof. 

\end{document}